\documentclass[thmsa,12pt]{article}
\normalfont\sffamily
\usepackage{amsfonts}
\usepackage{amssymb}
\usepackage[active]{srcltx}
\usepackage{amsmath}
\usepackage{amsthm}
\usepackage{graphicx}
\usepackage{textcomp}
\usepackage{a4wide}
\usepackage{chicago}
\usepackage{multirow}
\usepackage{color}
\usepackage[normalem]{ulem}
\usepackage{comment}
\usepackage{pxfonts}
\usepackage{setspace}
\usepackage{footmisc}

\usepackage{accents}
\usepackage{appendix}

\def\E{{\mathbf{E}}}

\renewcommand{\:}{\colon\!}
\def\l{\lambda}
\def\d{\delta}
\def\t{\theta}

\def\eps{\varepsilon}
\def\ptI{p_\t} 
\def\ptIl{\accentset{\sqsubset}{p_\t}} 
\def\ptIr{\accentset{\sqsupset}{p_\t}} 
\def\ptstrichI{p_{\t'}} 
\newcommand{\nmitte}[1]{k_{#1}}
\def\ntI{\nmitte{\t}}  
\newcommand{\nlinks}[1]{{ ^\triangleleft \mspace{-2mu} k}_{#1}}
\def\ntlinks{\nlinks{\t}} 
\newcommand{\nrechts}[1]{k^\triangleright_{#1}}
\def\ntrechts{\nrechts{\t}} 
\def\pperm{p(\ordering|\nvector)}  
\def\pvect{P(\nvector)} 
\def\nvector{\mathbf{k}} 
\def\sigext{t}
\def\median{\textsc{M}}
\def\ordering{\varrho}

\newcommand{\realnumbers}{\mathbb{R}}
\newcommand{\naturalnumbers}{\mathbb{N}}

{\begin{list}{}{%
\setlength{\leftmargin}{0.5cm}%
\setlength{\rightmargin}{0.5cm}}%
\item[]\ignorespaces}%
{\unskip\end{list}}

\newtheorem{theorem}{Theorem}
\newtheorem*{theoremiid*}{Theorem \ref{thm:main_theorem}}
\newtheorem*{theoremshock*}{Theorem \ref{thm:main_result_for_shock_case}}

\newtheorem{corollary}{Corollary}
\newtheorem{lemma}{Lemma}
\newtheorem*{lemma*}{Lemma \ref{lem:convergence_of_shocked_densities}'}

\begin{document}

\def\title #1{\begin{center}
{\Large {\sc #1}}
\end{center}}
\def\author #1{\begin{center} {#1}
\end{center}}

\setstretch{1.1}

\begin{titlepage}
\title{\sc On the Egalitarian Weights of Nations}

\author{Sascha Kurz\\ {\small Dept.\ of Mathematics, University of Bayreuth}}

\author{Nicola Maaser\\ {\small Dept.\ of Economics, University of Bremen}}

\author{Stefan Napel\\ {\small Dept.\ of Economics, University of Bayreuth;\\  Public Choice Research Centre, Turku}}

\begin{center} {\tt \today} \end{center}

\vspace{0.01cm}

\begin{center} {\bf {\sc Abstract}} \end{center}
{\small
Voters from $m$ disjoint constituencies (regions, federal states, etc.) are represented in an assembly which contains one delegate from each constituency and applies a weighted voting rule. All agents are assumed to have single-peaked preferences over an interval; each delegate's preferences match his constituency's \emph{median voter}; and the collective decision equals the assembly's \emph{Condorcet winner}.
We characterize the asymptotic behavior of the probability of a given delegate determining the outcome (i.e., being the weighted median of medians) in order to address a contentious practical question: which voting weights $w_1, \ldots, w_m$ ought to be selected if constituency sizes differ and all voters are to have \emph{a~priori equal influence} on collective decisions?
It is shown that if ideal point distributions have identical median $\median$ and are suitably continuous, the probability for a given delegate $i$'s ideal point~$\l_i$ being the Condorcet winner becomes asymptotically proportional to $i$'s \emph{voting weight} $w_i$ times $\l_i$'s density at $\median$ as $m\to \infty$.
Indirect representation of citizens is approximately egalitarian for weights proportional to the \emph{square root} of constituency sizes if all individual ideal points are i.i.d.
In contrast, weights that are \emph{linear} in~-- or, better, induce a \emph{Shapley value} linear in~-- size are egalitarian when preferences are sufficiently strongly affiliated within constituencies.
}

\vspace{0.2cm}

\begin{description}
{\small
\item[Keywords:]
two-tier voting systems; institutional design; collective choice; equal representation; Shapley value; pivot probability; quasivalues; voting power
\item[JEL codes:]
    D02; 
    D63; 
    D70; 
    H77 
}
\end{description}

\vspace{0.8cm}

\vfill
\noindent {\footnotesize Napel gratefully acknowledges financial support from the Academy of Finland through the Public Choice Research Centre, Turku. We thank Matthew Braham, Ulrich Kamecke, Mario Larch, Ines Lindner, Vincent Merlin, Abraham Neyman and Andreas Nohn for very helpful comments on earlier drafts. We have also benefitted from discussions with Robert Hable and feedback on seminar presentations in Aachen, Amsterdam, Augsburg, Bayreuth, Bilbao, Bozen, Caen, Freiburg, Hamburg, Maastricht, Munich, Paderborn, Salerno, Tilburg, Turku, Zurich and the Economic Theory committee of the Verein f\"{u}r Socialpolitik. The usual caveat applies.}

\end{titlepage}


\setstretch{1.19}

\section{Introduction}

The voting weights of delegations to electoral assemblies with a federal or divisional structure commonly vary in the size of the represented populations, but do so very differently. The US~Electoral College, for instance, involves voter blocs that are broadly proportional to constituency size: each state has two votes (reflecting its two seats in the Senate) in addition to a number which is proportional to population (like House seats).
California and Wyoming comprise around 11.9\% and 0.2\% of the US population, respectively, and thus end up holding around 10.2\% and 0.6\% of 
votes on the US~President. In contrast, the most and least populous member states of the EU~-- Germany and Malta~-- currently have about 8.4\% and 0.8\% of votes in the {Council of the European Union} but comprise 16.3\% and 0.1\% of the EU~population; the respective mapping from population size to voting weight is, very roughly, a square root function.\footnote{A least squares power-law regression of EU~Council voting weights $w_i$ on population sizes $n_i$ results in $w_i=c\cdot n_i^{0.47}$ with $R^2\approx 0.95$. The current Council voting rules 
involve two other but essentially negligible criteria, and will be changed in 2017 into a more proportional system.}  
Delegates in other collective decision-making bodies, such as the Governing Council of the European Central Bank, the Senate of Canada or German Bundesrat, and many~a university senate or council of a multi-branch NGO, have voting weights that are yet more concave functions of the number of represented constituents, or even flat.

This paper concerns \emph{two-tier voting systems} in which individuals vote on delegates or representatives in disjoint constituencies (bottom tier) and these representatives take collective decisions in a council, electoral college, or other assembly (top tier). It investigates a practically relevant, normative question: which simple function~-- possibly linear, possibly strictly concave or constant~-- \emph{should} determine the top-tier voting weights of delegates from differently sized constituencies such as US states or EU~member countries?
The considered objective is not one of
efficiently aggregating private information (see, e.g., \citeNP{Bouton/Castanheira:2012}) or of maximizing a utilitarian measure of welfare as investigated, for instance, by \citeN{Barbera/Jackson:2006}.
We focus on the egalitarian criterion of \emph{`one person, one vote'} and on providing all bottom-tier voters, at least a~priori and under very stylized ideal conditions, with \emph{equal influence} on the collective decision.
This is studied in a model where the respective {median voter} determines a constituency's top-tier policy position and the assembly's \emph{Condorcet winner} defines its collective decision. The relation of \emph{heterogeneity within each constituency} and \emph{heterogeneity across constituencies} turns out to be the critical determinant of the fair voting weight allocation. Linear and square root weighting rules emerge in particularly prominent benchmark cases. The former is advisable for electorates that are polarized along constituency lines, i.e., exhibit significant heterogeneity across constituencies; while the latter is more egalitarian when heterogeneity within each constituency is dominant.

The `one person, one vote' principle is linked to the requirement of \emph{anonymity} in social choice, that is, collective decisions shall depend only on the votes that the alternatives receive, not on whose votes these are. This general egalitarian norm is sometimes considered the minimum requirement for a decision-making procedure to be called `democratic' (e.g., \citeANP{Dahl:1956} \citeyearNP{Dahl:1956}, p.~37).
It is straightforward to implement~-- at least in theory~-- in case of a direct, single-tier voting procedure or a two-tier one with symmetric constituencies. Complications arise when a two-tier system is asymmetric. A non-trivial integer apportionment problem already needs to be resolved for those assemblies, like parliaments, in which delegates from the same constituency can split their votes and thus reflect heterogeneity among constituents (see \citeNP{Balinski/Young:2001}). And apportionment gets much more difficult when all representatives of a constituency vote as a bloc (as in the US~Electoral College, with two exceptions) 
or, equivalently, when the assembly contains a single delegate from each constituency who is endowed with a voting weight that varies in population size (as in the EU~Council).

A way to adapt the principle to such situations has very vaguely been suggested by the US Supreme Court, requiring ``that each citizen have an \emph{equally effective} voice in the election'' (cf.\ \emph{Reynolds v.\ Sims}, 377 U.S.~533, 1964, p.~565; emphasis added by the authors).
Here, we operationalize equal efficacy or influence by comparing the a~priori probabilities of individual voters being decisive or pivotal for the collective decision. The corresponding joint event of (i) a given voter determining her delegate's vote and of (ii) this representative determining the assembly's collective decision is admittedly a rare one. Still, while all being close to zero, the resulting individual \emph{pivot probabilities} can vary widely across constituencies when weights are chosen arbitrarily. They should not if an institutional designer wants to fix voting weights (or bloc sizes) which are fair at least from behind the constitutional `veil of ignorance'~-- that is, when preference patterns of the day are ignored for practical or normative reasons.

The objective of equalizing the a~priori influence of each citizen on collective decisions was first formally considered by Lionel S.\ \citeANP{Penrose:1946} in 1946, when the institutional design of a successor to the League of Nations~-- today's United Nations Organization (UNO)~-- was being discussed.\footnote{Informal investigations date back to anti-federalist writings by Luther Martin, a delegate from Maryland to the Constitutional Convention in Philadelphia in 1787. See \citeN{Riker:1986}. } \citeN{Penrose:1946} showed that the most intuitive solution to the weight allocation problem, i.e., weights proportional to constituency sizes, ignores ``elementary statistics of majority voting''.
Namely, if there are only two policy alternatives (`yes' and `no') and all individual decisions are statistically \emph{independent and equiprobable} then the probability of an individual voter being pivotal in her constituency with $n_i$ voters, which for odd $n_i$ corresponds to the probability of $n_i-1$ voters being divided into `yes' and `no'-camps of same size, is approximately $\sqrt{2}/\sqrt{\pi n_i}$ (apply Stirling's formula when evaluating the binomial distribution function).
So a voter from a constituency~$\mathcal{C}_i$ which is four times larger than constituency~$\mathcal{C}_j$ a~priori faces a smaller probability of tipping the scales locally; but this probability is still half rather than only a quarter of the reference one. Consequently, top-tier voting weights should be such that the pivot probability of constituency~$\mathcal{C}_i$ at the top tier is twice~-- not four times~-- that of $\mathcal{C}_j$ in order to equalize the indirect influence of all citizens.

The corresponding practical suggestion is also known as the \emph{Penrose square root rule}. Despite criticism that it treats voting decisions too much like coin tosses, the rule has provided a benchmark for numerous applied studies which consider the distribution of voting power in the EU, US, or IMF (including Felsenthal and Machover \citeyearNP{Felsenthal/Machover:2001}, \citeyearNP{Felsenthal/Machover:2004}; \citeNP{Grofman/Feld:2005}; 
\shortciteNP{Fidrmuc/Ginsburgh/Weber:2009}; \citeNP{Leech/Leech:2009}; \citeNP{Miller:2009}, \citeyearNP{Miller:2012}; \citeNP{Kirsch/Langner:2011}).
And though practitioners may not care about Penrose's reasoning itself~-- for instance, when the EU heads of state and government bargained on new, post-2017 voting rules for the Council~-- they have invoked Penrose's suggestion when it fitted their interests.\footnote{A particularly notorious case involved the then Polish president and prime minister 
in the negotiations of the Treaty of Lisbon. See, e.g., \citeN[June 14th]{Economist:2007}.}

The special role of square root weight allocation rules has been confirmed, qualified, and disputed in a number of studies on two-tier voting systems, both empirically (see Gelman et al.\ \citeyearNP{Gelman/Katz/Tuerlinckx:2002}; \citeyearNP{Gelman/Katz/Bafumi:2004}) and theoretically.
The respective constitutional objective functions and practical conclusions of these investigations vary.
Besides the equalization of a~priori influence (
\citeNP{Chamberlain/Rothschild:1981}; \citeNP{Felsenthal/Machover:1998}; \citeNP{Laruelle/Valenciano:2008}; \citeNP{Kaniovski:2008}), they consider utilitarian welfare maximization (e.g., \shortciteNP{Beisbart/Bovens/Hartmann:2005}; \citeNP{Barbera/Jackson:2006}; \citeNP{Beisbart/Bovens:2007}; \citeNP{Laruelle/Valenciano:2008}; \shortciteNP{Koriyama/Laslier/Mace/Treibich:2012}) and the avoidance of majoritarian paradoxes like having a Bush majority in the 2000~Electoral College despite a Gore majority in the population at large (\citeNP{Felsenthal/Machover:1999}; \citeNP{Kirsch:2007}; \shortciteNP{Feix/Lepelley/Merlin/Rouet/Vidu:2008}). Several departures from Penrose's independence and equiprobability assumptions have been considered.
However, the related literature has focused almost entirely on \emph{binary} political decisions, with no scope for bargaining and strategic interaction.\footnote{We are aware of the following exceptions only: Laruelle and Valenciano (\citeyearNP{Laruelle/Valenciano:2008:neutral}) suggest a ``neutral'' top-tier voting rule when policy alternatives give rise to a Nash bargaining problem. \shortciteN{LeBreton/Montero/Zaporozhets:2012} investigate fair voting weights in case of the division of a transferable surplus, i.e., for a simplex of policy alternatives. Maaser and Napel (\citeyearNP{Maaser/Napel:2007}; \citeyearNP{Maaser/Napel:2012:meanvoter}; \citeyearNP{Maaser/Napel:2012}) conduct simulations for a median voter environment like the one which we will consider here.}

The existing results hence provide useful guidance and arguments in thorny debates on the `right' weight allocation only to the extent that the assemblies in question indeed decide on dichotomous exogenous proposals. But many decisions
involve several shades of grey. Members of the US~Electoral College usually have binary options, but they face the survivors from a much larger field of initial contenders, with partly endogenous final political platforms. The EU~Council more commonly decides on the level of subsidies, the scope of regulation, the scale of financial aid, etc.\ rather than on having a subsidy, regulation of an industry, or aid \emph{per se}. It seems relevant, therefore, to analyze the fair choice of voting weights (and alternative objectives such as utilitarian welfare) for somewhat richer than binary $\{0,1\}$-policy spaces, too.

This paper considers the equalization of pivot probabilities for a one-dimensional convex policy space, i.e., for choices from a real \emph{interval}. We assume single-peaked preferences with random ideal points for all voters, perfect congruence between preferences of a constituency's delegate and its \emph{median voter}, and collective decisions which correspond to the \emph{Condorcet winner} or the \emph{core} of the game defined by preferences and weights of the delegates. The latter can be seen as the equilibrium outcome of strategic bargaining (see, e.g., \citeNP{Banks/Duggan:2000}).

In this model, the collective choice equals the weighted median among agents whose ideal points themselves are medians from disjoint samples. This is a very stylized representation of democratic decision making but yet richer than the binary model \`{a}~la Penrose. The former nests the latter in case that ideal points have a discrete two-point distribution. In case of less trivial distributions, little can analytically be said about the order statistics of medians from \emph{differently} sized samples; and next to nothing has so far been known about the combinatorial function therefrom which corresponds to the respective \emph{weighted median} of non-identically distributed random variables.

We here derive a general analytical result on the ratio of two delegates' pivot probabilities in an infinite increasing chain of collective decision bodies (Theorem~\ref{thm:main_theorem}). Each delegate~$i$ is characterized by his voting weight $w_i$ and single-peaked preferences with a random ideal point $\l_i$ that has the probability density function $f_i$. In line with the veil of ignorance perspective of constitutional design, this random variable is a~priori assumed to have the same theoretical median for all delegates~-- say, $\median=0$. It is shown then that, under suitable regularity conditions, a delegate's probability of finding his ideal point coincide with the corresponding voting game's Condorcet winner is \emph{asymptotically proportional to the probability density $f_i(\median)$ at the theoretical median times his assigned voting weight $w_i$}.

This main analytical result has several practical corollaries for two-tier voting systems. In particular, if all individual voters are~-- behind the constitutional veil of ignorance~-- conceived of as having ideal points that are \emph{independent and identically distributed (i.i.d.)}, then the sample median from a constituency $\mathcal{C}_i$ with $n_i$ members has an asymptotically normal distribution whose standard deviation is inversely proportional to the \emph{square root} of $n_i$. The probability density of representative~$i$'s ideal point at the theoretical median $\median$ is hence proportional to $\sqrt{n_i}$.
It follows that voting weights $w_i$ that are chosen to be proportional to $\sqrt{n_i}$ for all constituencies $\mathcal{C}_i$ render the top-tier pivot probabilities of all representatives proportional to their population sizes; this approximately equalizes the expected influence or efficacy of the vote across the population.
How close one gets to full equalization depends on the considered number $m$ of constituencies as well as the population partition at hand.\footnote{The approximation can be improved if one bases the weight choice on the induced \emph{Shapley value} or, with comparable effects, the \emph{Penrose-Banzhaf power index}. See \citeN{Dubey/Shapley:1979}, \citeN{Felsenthal/Machover:1998} or \citeN{Laruelle/Valenciano:2008} for good overviews on these and other power measures. \label{fn:SSI_and_PBI_introductions}}

The optimality of a square root allocation of voting weights is, however, restricted to the case of individual ideal points being i.i.d.\ and the use of a 50\%-majority threshold. Assuming that voters' ideal points are subject to identical random shocks within constituencies~-- which introduces positive correlation among members of the same constituency~-- implies greater variance of the respective sample medians.
The latter's distributions become more and more similar across constituencies if the shock distribution $H$ is identical for all constituencies $\mathcal{C}_i$ and its variance $\sigma_H^2$ increases. If this measure $\sigma_H^2$ of \emph{heterogeneity across constituencies} is sufficiently great relative to the \emph{heterogeneity within each constituency}, which is captured by the variance $\sigma_G^2$ of the 
(conditional) ideal point distribution $G$ under a zero shock, then an approximately \emph{linear} weight allocation becomes optimal.

That a linear weighting rule is optimal in this case follows as a corollary from Theorem~\ref{thm:main_theorem} when top-tier decisions are taken by simple majority. But, as made precise by  Theorem~\ref{thm:main_result_for_shock_case}, the finding can be extended to supermajority requirements.
In particular, one can approximate the pivot probabilities of top-tier delegates by the \emph{Shapley value} of the respective weighted voting game when $\sigma_H^2/\sigma_G^2$ is sufficiently large. Even if the number of constituencies $m$ is relatively small, one can hence achieve equal representation by finding voting weights such that the resulting Shapley value is proportional to population sizes, or as close to being proportional as is feasible.

The remainder of the paper is organized as follows. In Section~\ref{sec:Model}, we spell out our model of two-tier decision making and the institutional design problem. Our main result for simple majority rule and $m\to \infty$, as well as its corollary in case that individual ideal points are i.i.d.\ are presented in Section~\ref{sec:LimitAnalysis}.
We then explore the effect of adding heterogeneity across constituencies to that within, and study asymptotic behavior with respect to the `across'-kind for fixed $m$ in Section~\ref{sec:Heterogeneity}. We conclude in Section~\ref{sec:Conclusion} and provide proofs of the two theorems in an appendix.

\section{Model and Design Problem}\label{sec:Model}

We consider partitions $\mathfrak{C}^{m} = \{\mathcal{C}_1,\ldots,\mathcal{C}_m\}$ of a large number $n$ of voters into $m<n$ disjoint \emph{constituencies} with $n_i=|\mathcal{C}_i|>0$ members each. The preferences of any voter $l\in \{1, \ldots, n\}=\bigcup_i \mathcal{C}_i$ are assumed to be single-peaked 
with \emph{ideal point} $\nu^l$ in a convex one-dimensional \emph{policy space} $X\subseteq \realnumbers$, i.e., in a finite or infinite real interval.
These ideal points are conceived of as realizations of random variables with a~priori identical, absolutely continuous distributions. A given profile $(\nu^1, \ldots, \nu^n)$ of ideal points is interpreted as reflecting voter preferences in an abstract left--right spectrum or on a specific one-dimensional policy issue (a transfer, an emission standard, a capital requirement, etc.).

A collective decision $x^*\in X$ on the issue at hand is taken by an assembly or council of representatives $\mathcal{R}^m$ which consists of one representative from each constituency. Without committing to any particular procedure for internal preference aggregation, political competition, lobbying or bargaining, it will be assumed that preferences of $\mathcal{C}_i$'s representative coincide with those of its respective median voter, i.e., representative~$i$ has the random ideal point
\begin{equation}\label{eq:l_i_is_C_is_median}
\l_i\equiv \textnormal{median\,}\{\nu^l\colon l\in \mathcal{C}_i\}.
\end{equation}
For simplicity we take all $n_i$ as odd numbers,\footnote{For an even number $n_i$, one could let each of the two middlemost ideal points in $\mathcal{C}_i$ define the representative~$i$'s preferences with equal probability. Or one works with the usual definition of the median, i.e., their arithmetic mean, and focuses on the probability of event $\{\partial x^*/\partial \nu^i>0\}$ rather than the~-- no longer equivalent~-- event $\{x^*=\nu^i\}$ in what follows. \citeN{Napel/Widgren:2004} discuss in detail how \emph{influence} in voting procedures can be quantified by outcome sensitivity measures like $\partial x^*/\partial \nu^i$. \label{fn:even_n_i}}
and leave aside agency problems 
or other reasons for why the preferences of a constituency's representative might not be congruent or at least sensitive to its median voter.\footnote{See, e.g., \citeN{Gerber/Lewis:2004} for empirical evidence on how the median voter and partisan pressures jointly explain legislator preferences, and for a short discussion of the related theoretical literature.
It is important to note that Theorems~\ref{thm:main_theorem} and \ref{thm:main_result_for_shock_case} will \emph{not} require (\ref{eq:l_i_is_C_is_median}) to hold~-- they only assume $\l_i$'s density to have certain properties.}

In the top-tier assembly $\mathcal{R}^m$, constituency $\mathcal{C}_i$ has voting weight $w_i\ge 0$. Any coalition $S\subseteq \{1, \ldots, m\}$ of representatives which achieves a combined weight
$\sum_{j\in S} w_j$ above
\begin{equation}\label{eq:qm_defined}
q^{m}\equiv 0.5 \sum_{j=1}^m w_j,
\end{equation}
i.e.,\ which has a \emph{simple majority} of total weight, is winning and can pass proposals to implement some policy $x\in X$.

Let $\cdot \colon\! m$ be the random permutation of $\{1, \ldots, m\}$ that makes $\l_{k\: m}$ the $k$-th leftmost ideal point among the representatives for any realization of $\l_1, \ldots, \l_m$ (that is, $\l_{k\: m}$ is the $k$-th order statistic). We will disregard the zero probability events of two or more constituencies having identical ideal points and define the random variable $P$ by
\begin{equation}\label{eq:P_defined}
P\equiv \min \Big\{j\in\{1,\ldots,m\}\colon \sum_{k=1}^j w_{k:m} > q^{m}\Big\}.
\end{equation}
Representative $P\:m$'s ideal point, $\lambda_{P\:m}$, cannot be beaten by any alternative $x\in X$ in a pairwise vote, i.e., it is in the \emph{core} of the voting game defined by ideal points $\l_1, \ldots, \l_m$, weights $w_1,\ldots, w_m$ and quota $q^{m}$. We assume that the policy $x^*$ agreed by $\mathcal{R}^{m}$ lies in the core. So $x^*$ must equal $\lambda_{P\:m}$ whenever the core is single-valued;
then $\lambda_{P\:m}$ actually beats every other alternative $x\in X$ and is the so-called \emph{Condorcet winner} in $\mathcal{R}^m$.
In order to avoid inessential case distinctions, we assume that $\mathcal{R}^{m}$ agrees on $\lambda_{P\:m}$ also in the knife-edge case of the entire interval $[\lambda_{P-1\:m}, \lambda_{P\:m}]$ being majority-undominated, i.e.,\footnote{A sufficient condition for the core to be single-valued is that the vector of weights satisfies $\sum_{j\in S} w_j\neq  q^{m}$ for each  $S\subseteq \{1,\ldots,m\}$. In the non-generic cases where this is violated, tie-breaking assumptions analogous to fn.~\ref{fn:even_n_i} can be made. Note that no constituency's median voter will have an incentive to `choose' a representative whose preferences differ from her own ones, that is, to misrepresent preferences, if $x^*$ is determined by (\ref{eq:outcome_equals_pivotal_point}) (cf.\ \citeNP{Moulin:1980}; \citeNP{Nehring/Puppe:2007}). }
\begin{equation}\label{eq:outcome_equals_pivotal_point}
x^*\equiv\l_{P\:m}.
\end{equation}
Representative $P\:m$ will, therefore, generally be referred to as the \emph{pivotal representative} or the \emph{weighted median} of $\mathcal{R}^{m}$. \citeN{Banks/Duggan:2000} and \citeN{Cho/Duggan:2009} provide equilibrium analysis of non-cooperative legislative bargaining which supports policy outcomes inside or close to the core.

The event $\{x^*=\nu^l\}$ of voter~$l$'s ideal point coinciding with the collective decision almost surely entails that sufficiently small perturbations or idiosyncratic shifts of $\nu^l$ translate into identical shifts of $x^*$, so that $\partial x^*/\partial \nu^l>0$. Voter~$l$ can then meaningfully be said to \emph{influence}, be \emph{decisive} or \emph{pivotal} for, or even to \emph{determine} the collective decision. This event has probability
\begin{equation}
p^l\equiv\Pr(x^*=\nu^l),
\end{equation}
which depends on the joint distribution of $(\nu^1, \ldots, \nu^n)$ and the voting weights $w_1, \ldots, w_m$ that have been selected for $\mathcal{R}^m$. Even though $p^l$ will be very small given that the set of voters $\{1, \ldots, n\}$ is assumed to be large, it would constitute a violation of the `one person, one vote' principle if $p^l/p^k$ differed substantially from unity for any $l,k\in\{1, \ldots, n\}$.

We will assume throughout our analysis that all voter ideal points are a~priori \emph{identically distributed}, in line with adopting a `veil of ignorance'-perspective when one analyzes the efficacy of individual votes or a~priori influence of voters. Moreover, it is assumed that ideal points are mutually \emph{independent across constituencies}. We do, however, allow for a specific form of ideal points being \emph{dependent within each constituency}.
Namely, we conceive of the ideal point $\nu^l$ of any voter $l\in \mathcal{C}_i$ as the sum
\begin{equation}\label{eq:nu_decomposed}
\nu^l=\mu_i+\epsilon^l
\end{equation}
of a constituency-specific random variable $\mu_i$, which has distribution $H$, and a voter-specific random variable $\epsilon^l$ with absolutely continuous distribution $G$.
The voter-specific variables $\epsilon^1, \ldots, \epsilon^n$ and constituency shocks $\mu_1, \ldots, \mu_m$ are all taken to be mutually independent.
If distribution $H$ of $\mu_i$ is non-degenerate, it reflects a common attitude component of preferences within the disjoint constituencies.
$G$ and $H$ are the same for all voters $l\in \{1, \ldots, n\}$ and constituencies $\mathcal{C}_i\in \mathfrak{C}^m$. This ensures that indeed all ideal points $\nu^1, \ldots, \nu^n$ are identically distributed. $G$'s variance $\sigma^2_G$ can be interpreted as a measure of \emph{heterogeneity within each constituency}, reflecting the natural variation of political preferences.
Similarly, $\sigma^2_H$ is a measure of \emph{heterogeneity across constituencies}: even though it is assumed that opinions in all constituencies vary between left--right, religious--secular, etc.\ in a similar manner, the locations of the respective ranges of opinion can differ between constituencies.
The corresponding correlation coefficient for two voters $l,k \in \mathcal{C}_i$ from the same constituency is $\sigma^2_H/(\sigma^2_H+\sigma^2_G)$. The case in which $H$ is degenerate with $\sigma^2_H=0$ involves heterogeneity only within constituencies; the latter differ in size but voter ideal points $v^l$ are independent and identically distributed across the entire population. We regard this as a particularly important benchmark and will refer to it as the \emph{i.i.d.~case}.

With this notation, we can now state our objective of operationalizing the `one person, one vote' principle somewhat more formally. Namely, given a partition $\mathfrak{C}^{m} = \{\mathcal{C}_1,\ldots,\mathcal{C}_m\}$ of $n$ voters into constituencies and distributions $G$ and $H$ which describe heterogeneity of individual preferences within and across constituencies, we would like to find voting weights $w_1, \ldots, w_m$ such that each voter a~priori has an equal chance of determining the collective decision $x^*\in X$~-- that is, such that
\begin{equation}\label{eq:OPOV_condition}
    \frac{p^l}{p^k}=1\, \textnormal{\ \emph{for all} } l, k\in \{1, \ldots, n\}.
\end{equation}

For most combinations of $\mathfrak{C}^{m}$, $G$, and $H$, condition  (\ref{eq:OPOV_condition}) cannot be satisfied by \emph{any} weight vector $(w_1, \ldots, w_m)$. This is due to the discrete nature of weighted voting.\footnote{For instance, there are only 117 structurally different weighted voting games with $m=5$ constituencies even if all majority thresholds between 0\% and 100\% are permitted. This number (related to \emph{Dedekind's problem} in discrete mathematics) grows very fast, but the set of distinct feasible influence distributions remains finite.} So the problem would need to be formulated more precisely as that of minimizing a specific notion of distance between the probability vector $(p^1, \ldots, p^n)$ induced by $w_1, \ldots, w_m$ and $(1/n, \ldots, 1/n)\in \realnumbers^n$.

The actual concern, however, is \emph{not} with finding the respective optimal solution to such a (non-trivial) discrete minimization problem for a particular partition $\mathfrak{C}^{m}$ and specific distributions $G$ and $H$.
Rather, our objective is to find a \emph{simple function} which maps $n_1, \ldots, n_m$ to weights $w_1, \ldots, w_m$ that induce ${p^l}/{p^k}\approx 1$ for all $l$ and $k$, that is, which approximately satisfy the `one person, one vote' criterion, for arbitrary non-pathological partitions $\mathfrak{C}^{m}$.\footnote{By pathological partitions we, for instance, mean ones where constituency sizes $n_i$ increase exponentially in $i$.~-- There is no need to specify exactly which functions are ``simple'' enough. \emph{Power laws}, i.e., choosing $w_i=\beta n_i^\alpha$ for some $\alpha,\beta\in \realnumbers$, certainly qualify and turn out to constitute a sufficiently rich class of mappings.}
Preferably, qualitative information on heterogeneity within and across constituencies should suffice in guiding possible design recommendations.

The stated assumptions imply that, when considering any given realization of $\mu_i$, ideal points $\nu^l$ and $\nu^k$ are conditionally independent if $l, k\in \mathcal{C}_i$ for some $i$. They are in any case identically distributed. In particular, $p^l=p^k$ holds for $l, k\in \mathcal{C}_i$ irrespective of which $G$, $H$, and voting weights $w_1, \ldots, w_m$ are considered, and it must be the case that if $l\in \mathcal{C}_i$ then
\begin{equation}\label{eq:within_constituency_pivot_probabilities}
    \Pr(\nu^l=\l_i)=\frac{1}{n_i}.
\end{equation}
So, an individual voter's probability to be her constituency's median and to determine $\l_i$ is \emph{inversely proportional} to constituency $\mathcal{C}_i$'s population size.

The events $\{\nu^l=\l_i\}$ and $\{x^*=\l_i\}$ are independent given our statistical assumptions. (Note that the first event only entails information about the identity of $\mathcal{C}_i$'s median, not its location.) It follows that the probability $p^l$ for an individual voter~$l\in\mathcal{C}_i$ influencing the collective decision $x^*$ is $1/n_i$ times the probability of event $\{x^*=\l_i\}$ or, equivalently, of $\{P\:m = i\}$.
Letting
\begin{equation}
\pi_i(\mathcal{R}^m)\equiv \Pr(P\:m = i)
\end{equation}
denote the probability of constituency $\mathcal{C}_i$'s representative being pivotal in $\mathcal{R}^m$ (that is, of $\l_i$ being the respective Condorcet winner in case of generic weights), our institutional design objective hence consists of solving the following
\begin{quotation}
\noindent \emph{Problem of Equal Representation}:\\
\emph{Find a simple mapping from constituency sizes $n_1, \ldots, n_m$ to voting weights $w_1, \ldots, w_m$ for the representatives in $\mathcal{R}^m$ such that}
\begin{equation}\label{eq:OPOV_condition_restated}
    \frac{\pi_i(\mathcal{R}^m)}{\pi_j(\mathcal{R}^m)}\approx \frac{n_i}{n_j} \textnormal{ \emph{ for all} } i, j\in \{1, \ldots, m\}.
\end{equation}
\end{quotation}
One might conjecture that, if $m$ is large enough and the weight distribution is not overly skewed, voting weight $w_i$ should translate linearly into representative~$i$'s influence $\pi_i(\mathcal{R}^m)$.\footnote{
Asymptotic proportionality between weights and voting power has first been investigated  by \citeN{Penrose:1952} in the context of binary alternatives. Related formal results by \citeN{Neyman:1982}, \citeN{Lindner/Machover:2004}, \shortciteN{Snyder/Ting/Ansolabehere:2005}, \citeN{Jelnov/Tauman:2012} and Theorem~\ref{thm:main_theorem} below suppose that the \emph{relative} weight of any given voter becomes negligible as more and more voters are added. The case when relative weights of a few large voters fail to vanish as $m\to \infty$~-- giving rise to \emph{oceanic games} and typically non-proportionality~-- has been treated by \citeN{Shapiro/Shapley:1978} and \citeN{Dubey/Shapley:1979}. The limit behavior of pivot probabilities for \emph{uniform} weights (as at the bottom tier) has been studied in more complex models than Penrose's by \citeN{Chamberlain/Rothschild:1981}, \citeN{Myerson:2000}, \shortciteN{Gelman/Katz/Tuerlinckx:2002}, and \citeN{Kaniovski:2008}.}
But the distribution of the respective ideal points $\l_i$ will certainly play a role, too, and so the solution to this problem will depend on how heterogeneity of individual preferences within and across constituencies relate.

Note that, \emph{if} the representatives' ideal points $\l_1, \ldots, \l_m$ were not only mutually independent but also had identical distributions $F_i=F_j$ for all $i,j\in \{1, \ldots, m\}$, then all orderings of $\l_1, \ldots, \l_m$ would a~priori be equally likely.
In this case, $\pi_i(\mathcal{R}^m)$ would coincide with $i$'s \emph{Shapley value} $\phi_i(v)$, where $v$ is the characteristic function of the $m$-player cooperative game in which the worth $v(S)$ of a coalition $S\subseteq \{1, \ldots, m\}$ is 1 if $\sum_{j\in S} w_j>q^m$ and 0 otherwise,
and\footnote{See \citeN{Shapley:1953}. For so-called \emph{simple games}~-- in which $v(S)\in \{0,1\}$ for all $S\subseteq N$, $v(\varnothing)=0$, $v(N)=1$, and $v(S)=1\Rightarrow v(T)=1$ if $S\subseteq T$~-- the Shapley value is also referred to as the \emph{Shapley-Shubik power index}, following the first suggestion of using $\phi$ in order to evaluate power in voting bodies by \citeN{Shapley/Shubik:1954}. We write $v=[q^m;w_1, \ldots, w_m]$ if $v$ is defined by the weighted voting rule $[q^m;w_1, \ldots, w_m]$.}
\begin{equation}
\phi_i(v)\equiv \sum_{S\subseteq \{1, \ldots, m\}\smallsetminus \{i\}} \frac{|S|!\cdot(m-|S|-1)!}{m!}[v(S\cup \{i\})-v(S)].
\end{equation}

The way to solve the problem of equal representation would then simply be to search for a weighted voting game which induces a Shapley value proportional to $(n_1,\ldots,n_m)$. 
Unfortunately, if we assume that distribution $G$, which generates the private component $\epsilon^l$ in individual ideal points $\nu^l$, is non-degenerate then (\ref{eq:l_i_is_C_is_median}) implies that $F_i=F_j$ if and only if $n_i=n_j$.
The case in which this holds for all $i,j\in \{1, \ldots, m\}$ is precisely the one in which equal representation is trivial, i.e., achieved by giving all representatives  identical weights because $n_1=\ldots=n_m$.\footnote{We remark that re-partitioning the population into constituencies of equal size~-- i.e., appropriate redistricting~-- is, of course, a trivial possibility for altogether evading the considered problem. Our analysis is concerned with those cases where historical, geographical, cultural, and other reasons exogenously have defined a partition $\mathfrak{C}^m$ which cannot easily be changed.
See \citeN{Coate/Knight:2007} on socially optimal districting and \citeN{Gul/Pesendorfer:2010} on strategic issues which arise for redistricting. We also disregard another relevant strategic feature of two-tier voting: incentives to allocate limited campaign resources to the constituencies. We refer the reader to 
\citeN{Stromberg:2008}.
}
In particular, $F_i$ second-order stochastically dominates distribution $F_j$ (or $F_j$ is a mean-preserving spread of $F_i$) if $n_i>n_j$ because the sample median of $n_i$ independent draws from $G$ has smaller variance than that of just $n_j<n_i$ draws; and the respective draw from $H$ adds identical variance to both $\l_i$ and $\l_j$.

In the i.i.d.\ benchmark case, in which $\sigma^2_H=0$ and the only acknowledged differences between two voters from distinct constituencies are the numbers of their fellow constituents, one can be more specific than stochastic dominance. Namely, a standard result about the sample median of i.i.d.\ random variables is:
\begin{lemma}\label{lem:Arnold_et_al}
Let $X_1, \ldots, X_s$ be i.i.d.\ random variables with median $\median$ and a density $f$ that is continuous at $\median$ with $f(\median)>0$. Then random variable $Y=\textnormal{median\,}\{X_1, \ldots, X_s\}$ is asymptotically $(\median,\sigma^2)$-normally-distributed with
\begin{equation}\label{eq:variance_of_median}
{\sigma^2}=\frac{1}{s\,[2f(\median)]^2}
\end{equation}
i.e., the re-scaled sample median $2f(\median)\sqrt{s}[Y-\median]$ of $X_1, \ldots, X_s$ converges in distribution to $\mathbf{N}(0,1)$ as $s\to \infty$.

\end{lemma}
\noindent See, e.g., \shortciteN[Theorem~8.5.1]{Arnold/Balakrishnan/Nagaraja:1992} for a proof. It follows that in the i.i.d.\ case, ideal points $\l_1, \ldots, \l_m$ in assembly $\mathcal{R}^m$ are (approximately) normally distributed\footnote{This approximation is very good already for rather moderate sample sizes. If, e.g., individual ideal points $\nu^l$ are standard uniformly distributed, i.e., $\epsilon^l\sim \mathbf{U}[0,1]$ and $\mu_i\equiv 0$, then $\l_i$ is beta distributed with parameters $a=b=(n_i+1)/2$. The corresponding beta and normal density functions can be regarded as identical for all practical purposes if $n_i>100$.~--
Note that Lemma~\ref{lem:Arnold_et_al} is useful also in case of a non-degenerate distribution $H$ of $\mu_i$: it establishes that the precise distribution $G$ of individual shocks does not matter for $\l_i$'s distribution $F_i$; only $g(\median)$, the sufficiently great $n_i$, and $H$ do.}
with identical means but standard deviations that are \emph{inversely proportional to the square root} of the respective constituency sizes.
So, in the i.i.d.\ case, rather than all orderings of $\l_1, \ldots, \l_m$ being equally likely ($\phi(v)$'s implicit assertion), the representative of a constituency $\mathcal{C}_i$ which is four times larger than constituency $\mathcal{C}_j$ has twice the chances to find itself in the middle. (Recall that normal density at the mean and median is inversely proportional to standard deviation.)
Then weights that are proportional to population sizes, or weights such that $\phi(v)$ is, would give representatives of large constituencies more a~priori influence than is due.\footnote{We are unaware of any systematic empirical evidence for or against the claim that representatives from larger constituencies tend to be located more centrally in the relevant policy space. This theoretical prediction is testable in principle but rests on the two assumptions of aggregate preferences being determined by the median individual \emph{and} individual ideal points being i.i.d.
In view of the rapid transition towards indistinguishable distributions of representative ideal points when the i.i.d.\ assumption is given up (see Section~\ref{sec:Heterogeneity}), the claim is bound to be difficult to confirm in practice.}

Before we make this reasoning precise in the following section, let us iterate that the considered median voter model of equal representation in two-tier decision making is an admittedly big simplification. Many collective decisions involve more than just a single dimension in which voter preferences differ.
Even if the assumption of a one-dimensional, say, left--right policy space and single-peaked preferences was granted, systematic abstention of certain social groups could drive a wedge between the median voter's and the median citizen's preferences.\footnote{Parts of the population may be without suffrage (minors, aliens, or prisoners). \citeN[p.~57]{Penrose:1946} referred to ``the square root of the number of people on each nation's voting list'' but the political discussion of voting weights in the EU~Council has almost exclusively referred to \emph{population figures}. We here follow this line and use the terms citizens and voters interchangeably.}
We ignore that voting might involve private information about some state variable (\citeNP{Feddersen/Pesendorfer:1996}, \citeyearNP{Feddersen/Pesendorfer:1997}; \citeNP{Bouton/Castanheira:2012}), and typical agency problems connected to imperfect monitoring and infrequent delegate selections (for instance, national elections of the EU Council's members take place every four to five years).
Empirical evidence highlights that a representative may take positions that differ significantly from his district's median when voter preferences within that district are sufficiently heterogeneous (see, e.g., \citeNP{Gerber/Lewis:2004}).
Still, we take it that the best intuitions about fairness are captured by simplifying thought experiments of a `veil of ignorance' kind.
The  analysis of the described stylized world~-- no friction, particularly well-behaved preferences which are a~priori identical for all~-- is useful in this way.
It shows the limitations of and justifications for the simple intuition that weights should be proportional to the number of represented constituents, in a framework that goes beyond the binary world analyzed by \citeN{Penrose:1946} and others.

\section{Egalitarian Voting Weights for Many Constituencies}\label{sec:LimitAnalysis}

We will in this section consider situations in which the number $m$ of constituencies is suitably large. Very few tangible results exist on the distribution of order statistics, like the median, from \emph{differently} distributed random variables (the representative ideal points $\l_1, \ldots, \l_m$). And almost nothing seems to be known about the respective distribution of a \emph{weighted median}, which is taken to define the collective decision $x^*\in X$ in our model. It turns out to be possible, nevertheless, to characterize the probability of some $\l_i$ being the weighted median, i.e., the pivot probability $\pi_i(\mathcal{R}^m)$, as $m\to\infty$.

We conceive of $\mathcal{R}^{1}\subset \mathcal{R}^{2} \subset \mathcal{R}^{3}\subset \ldots$ as an infinite chain of assemblies in which more and more constituencies $i\in \naturalnumbers$ have a representative with a voting weight $w_i\ge 0$ and a random ideal point $\l_i$ with absolutely continuous distribution $F_i$.
Some technical requirements will be imposed on the corresponding density $f_i$, but it does not matter if $\l_i$ is defined by (\ref{eq:l_i_is_C_is_median}) and corresponds to the median of some set of other random ideal points like $\{\nu^l\}_{l\in \mathcal{C}_i}$; it could, for instance, be the average of some ideal points (such as those of members of a coalition government or oligarchy) or that of a constituency dictator.
So while the  problem of equal representation which we stated in Section~\ref{sec:LimitAnalysis} is the key motivation for investigating pivot probabilities $\pi_1(\mathcal{R}^m),\ldots,\pi_m(\mathcal{R}^m)$, the following characterization of their limiting behavior has more general applicability.\footnote{In particular, for given $F_1, \ldots, F_m$, $\pi(\mathcal{R}^m)$ amounts to a specific \emph{quasivalue} or \emph{random order value} for simple games. See, e.g., \citeN{Monderer/Samet:2002}.}
We will return to the issue of designing egalitarian two-tier voting systems after considering assemblies $\mathcal{R}^m$ with rather arbitrary ideal point distributions $F_1, \ldots, F_m$ and weighted voting rules $[q^m; w_1, \ldots, w_m]$.

For weight sequences $\{\boldsymbol{w}^m\}_{m\in \naturalnumbers}$ and associated weighted voting games $[q^m; w_1,\ldots, w_m]$ in which nobody's \emph{relative} voting weight is bounded away from zero, the pivot probability $\pi_i(\mathcal{R}^{m})$ of any given representative $i\in \naturalnumbers$ will converge to zero as $m\to \infty$.\footnote{Concerning the Shapley value $\phi(v)$, which equals $\pi(\mathcal{R}^{m})$ if $F_1, \ldots, F_m$ are identical, \citeN[Lemma~8.2]{Neyman:1982} has established that $\phi_i(v)\le 4w_i/\sum_{j>i} w_j$ for $v=[q; w_1, \ldots, w_m]$ with $w_1\ge \ldots \ge w_m\ge 0$, $\sum_{j=1}^m w_j=1$, and $2w_1<q<1-2w_1$.}
Still, $\pi_i(\mathcal{R}^{m})/\pi_j(\mathcal{R}^{m})$ need not converge.
This is illustrated by the sequence $\{\boldsymbol{w}^m\}_{m\in \naturalnumbers}$ with
\begin{equation}\label{eq:1222_example}
\boldsymbol{w}^m=(1, 2, \ldots, 2)\in \realnumbers^m.
\end{equation}
Representative~1 is either a dummy player with $\pi_1(\mathcal{R}^{m})=0$ or, supposing that the ideal point distributions $F_1, \ldots, F_m$ are identical, $\pi_i(\mathcal{R}^{m})=\frac{1}{m}$ for \emph{all} $i=1, \ldots, m$ depending on whether $m$ is odd or even. So $\pi_1(\mathcal{R}^{m})/\pi_2(\mathcal{R}^{m})$ alternates between 0 and 1. More complicated examples of non-convergence can be constructed, e.g., by having $\{\boldsymbol{w}^m\}_{m\in \naturalnumbers}$ oscillate in a suitable fashion.

We rule out such possibilities by imposing a weak form of replica structure on the considered weights $w_1, w_2, w_3, \ldots$ and ideal point distributions $F_1, F_2, F_3, \ldots$\,\ Specifically, we require that all representatives $i\in \naturalnumbers$ belong to one of an arbitrary but finite number $r$ of representative \emph{types} $\t\in \{1, \ldots, r\}$. All representatives of the same type $\t$ have an identical weight $w_\t$ and distribution $F_\t$. And, avoiding somewhat contrived situations like in (\ref{eq:1222_example}), we restrict attention to chains $\mathcal{R}^{1}\subset \mathcal{R}^{2} \subset \mathcal{R}^{3}\subset \ldots$ in which each type $\t$ maintains a non-vanishing share of representatives as $m\to \infty$.

The key requirements for the following result are that (i) $F_1, F_2, F_3 \ldots$ have identical median $\median$ and that (ii) each distribution $F_i$ has a density $f_i$ which is locally continuous and positive at $\median$. In order to allow the application of a powerful uniform convergence result for the Shapley value by \citeN{Neyman:1982}, continuity will be strengthened to the requirement that each density $f_i$'s variation at its median, $|f_i(x)-f_i(\median)|$, can locally be bounded by a quadratic function $cx^2$. This bound follows readily if $f_i$ is $C^2$ like the normal density functions singled out by Lemma~\ref{lem:Arnold_et_al}, and could be relaxed to $cx^a$ for any $a>1$ if one used somewhat less round constants in the proof. Moreover, an unpublished extension by Abraham Neyman of his 1982 result could be employed in order to make do with just (ii). Details on this and the proof are presented in Appendix~A.

\begin{theorem}\label{thm:main_theorem}
Consider an infinite chain $\mathcal{R}^{1}\subset \mathcal{R}^{2} \subset \mathcal{R}^{3}\subset \ldots$ of assemblies which involves a finite number $r$ of representative types, i.e., there exists a mapping $\tau\colon \naturalnumbers \to \{1, \ldots, r\}$ such that $\tau(j)=\t$ implies that $\l_j$ has density $f_\t$ and $w_j=w_\t\ge 0$.\footnote{We presume w.l.o.g.\ that $\tau(i)=i$ for $i\in \{1, \ldots, r\}$.}
Let the share of each type be bounded away from zero, i.e.,  there exist $\beta>0$ and $m^0\in \naturalnumbers$ such that $\beta_\t(m)\equiv |\mspace{1.5mu}\{k\in \{1, \ldots, m\}\colon \tau(k)=\t\}\mspace{2mu}|\mspace{2mu}/m \geq \beta>0$ for all $m\geq m^0$.
If for each $\t\in \{1, \ldots, r\}$ the distribution $F_\t$ has median $\median$ and its density $f_\t$ satisfies $f_\t(\median)>0$ with $|f_\t(x)-f_\t(\median)|\leq cx^2$ on a non-empty interval $[\median-\eps_1,\median+\eps_1]$ for some $c\geq 0$ then for $w_j>0$
\begin{equation}\label{eq:thm}
\lim_{m\to \infty}\frac{\pi_i(\mathcal{R}^{m})}{\pi_j(\mathcal{R}^{m})} = \frac{w_i f_{i}(\median)}{w_j f_{j}(\median)}.
\end{equation}
\end{theorem}

The key observation behind limit result (\ref{eq:thm}) and its corollaries is that, as $m$ grows large, the pivotal member of $\mathcal{R}^{m}$ is most likely found very close to the common median $\median$ of ideal point distributions $F_1, \ldots, F_m$. Pivotality at location $x\in X$ requires that less than half the total weight of $\mathcal{R}^{m}$'s members is located in $(-\infty, x)$ and less than half the total weight is found in $(x,\infty)$.
In expectation, this occurs exactly at $x=\median$ and, by Hoeffding's inequality, the probability for the realized weighted median in $\mathcal{R}^{m}$ to fall outside an $\eps$-neighborhood of $\median$ approaches zero exponentially fast as $m\to \infty$.

One can, therefore, restrict attention to an arbitrarily small interval $[-\eps,\eps]\subset[-\eps_1,\eps_1]$ for sufficiently large $m$ if w.l.o.g.\ $\median=0$. Since the densities $f_1(x), \ldots, f_m(x)$ satisfy the mentioned kind of continuity, they can suitably be approximated by upper and lower bounds on this interval. Moreover, when we condition on the respective events $\{\l_j \in [-\eps,\eps]\}$, the corresponding bounds are almost identical for any $j=1, \ldots, m$ when $m$ is sufficiently large.
This makes all orderings of those representatives with ideal points in $[-\eps,\eps]$ conditionally equiprobable in very good approximation. Representative~$i$'s respective conditional pivot probability, therefore, corresponds to $i$'s Shapley value in a `subgame' which involves only the representatives $j$ with realizations $\l_j\in [-\eps,\eps]$.
It is possible to apply the uniform convergence result for the Shapley value proven by \citeN{Neyman:1982} to each of these subgames.
In the final step of the proof, it then remains to exploit that the probability of the condition $\{\l_i \in [-\eps,\eps]\}$ being true becomes proportional to $\l_i$'s density at 0 when $\eps\downarrow 0$.\footnote{
At an intuitive level, one may even directly think of the limit case $\eps=0$: if one conditions on representative $i$'s ideal point being located at $x=\median=0$, i.e., $\{\l_i=0\}$, then each representative $j\neq i$ is equally likely found to $i$'s left or right (because $F_j(\median)=\frac{1}{2}$).
In this case, $i$'s conditional pivot probability equals $i$'s Penrose-Banzhaf power index, which, like the Shapley value, becomes proportional to $(w_1,\ldots, w_m) $ for the replica-like weight sequences that we consider. (See \citeNP{Lindner/Machover:2004} and \citeNP{Lindner/Owen:2007} on the corresponding limit result.)}

Theorem~\ref{thm:main_theorem} provides a rather general answer to the posed problem of equal representation in the case that many constituencies are involved. In particular, comparison of (\ref{eq:OPOV_condition_restated}) and (\ref{eq:thm}) immediately yields
\begin{corollary}\label{cor:general_solution_simple}
If $m$ is sufficiently large then choosing
\begin{equation}
(w_1, \ldots, w_m)\propto \left(\frac{n_1}{f_1(\median)},\ldots, \frac{n_m}{f_m(\median)}\right)
\label{eq:general_solution_simple}
\end{equation}
achieves approximately equal representation (as formalized by condition~(\ref{eq:OPOV_condition_restated})) if the technical conditions of Theorem~\ref{thm:main_theorem} are verified by $f_1, \ldots, f_m$.
\end{corollary}
If $m\ll \infty$, the approximation of the conditional pivot probabilities for ideal points in a neighborhood of the common median $\median$, which is obtained by (a) considering the limit case of orderings being conditionally equiprobable and (b) by applying Neyman's limit result for the Shapley value, need not be very good. The latter source of imprecision can be avoided by computing the Shapley value $\phi(v)$ for the simple game $v=[q^m;w_1,\ldots, w_m]$ which is defined by representatives' weights and simple majority rule as described in Section~\ref{sec:Model}. The suggestion in Corollary~\ref{cor:general_solution_simple} can hence be improved somewhat if (\ref{eq:general_solution_simple}) is replaced by
\begin{equation}\label{eq:general_solution_SSI_based}
(w_1, \ldots, w_m) \textnormal{\ \,such that\ \,} \phi(q^m;w_1,\ldots, w_m)\propto \left(\frac{n_1}{f_1(\median)},\ldots, \frac{n_m}{f_m(\median)}\right).
\end{equation}

We conclude this section by specifically considering the benchmark i.i.d.\ case, in which the ideal points $\nu^1, \ldots, \nu^n$ correspond just to voter-specific random variables $\epsilon^1, \ldots, \epsilon^n$ that are independent and identically distributed with a suitable probability density function $g$, and where $\l_i$ corresponds to the median ideal point in $\mathcal{C}_i$. In this case, the ideal points $\l_1, \ldots, \l_m$ in assembly $\mathcal{R}^m$ are asymptotically normally distributed by Lemma~\ref{lem:Arnold_et_al} with respective densities that satisfy the quadratic bound condition of Theorem~\ref{thm:main_theorem} and
\begin{equation}\label{eq:iid_median_density}
f_i(\median)=\frac{1}{\sqrt{2\pi\cdot\frac{1}{ n_i[2g(\median)]^2}}} =
\frac{2g(\median)}{\sqrt{2\pi}}\sqrt{n_i}>0.
\end{equation}
Combining (\ref{eq:iid_median_density}) and Corollary~\ref{cor:general_solution_simple} we obtain:
\begin{corollary}[Square root rule]\label{cor:equal_representation_iid}
If the ideal points of all voters are i.i.d., representative~$i$'s ideal point equals the median voter's ideal point in constituency $\mathcal{C}_i$ for all $i\in \{1, \ldots, m\}$, and $m$ is sufficiently large then
\begin{equation}
(w_1, \ldots, w_m)\propto \left(\sqrt{n_1},\ldots, \sqrt{n_m}\right)
\label{eq:iid_solution_simple}
\end{equation}
or, better,
\begin{equation}
(w_1, \ldots, w_m) \textnormal{\ \,such that\ \,} \phi(q^m;w_1,\ldots, w_m) \propto \left(\sqrt{n_1},\ldots, \sqrt{n_m}\right)
\label{eq:iid_solution_SSI_based}
\end{equation}
achieves approximately equal representation. 
\end{corollary}

\section{Heterogeneity within vs.\ across constituencies}\label{sec:Heterogeneity}

Corollary~\ref{cor:equal_representation_iid} derived a \emph{square root} rule similar to that of \citeN{Penrose:1946}
for the i.i.d.\ case,\footnote{Note that the \emph{Penrose square root rule} does not refer to weights but top-tier pivot
probabilities, which equal the Penrose-Banzhaf power index of the representatives in Penrose's binomial voting model
(cf.\ fn.~\ref{fn:SSI_and_PBI_introductions}).} that is, for a degenerate distribution $H$ of the constituency-specific
$\mu_i$-components of individual ideal points $\nu^l=\mu_i+\epsilon^l$. 
We now investigate the robustness of this rule regarding the degree of preference affiliation within each constituency. Non-degenerate shocks $\mu_i$ imply positive correlation within each constituency and give rise to \emph{polarization of preferences} along constituency lines, which is measured by the ratio $\sigma_H^2/\sigma_G^2$.\footnote{The basic features of \emph{polarization} are according to \citeN[p.~824]{Esteban/Ray:1994}: (i) a high degree of homogeneity \emph{within} groups, (ii) a high degree of heterogeneity \emph{across} groups, and (iii) a small number of significantly sized groups. Ratio $\sigma^2_H/\sigma^2_G$ serves as a simple measure of polarization of ideal points $\nu^1, \ldots, \nu^n$ here, where groups are given exogenously. \citeANP{Esteban/Ray:1994} characterize polarization measures for the general case without an exogenous partition of the population.}
It turns out that for sufficiently strong polarization, 
a \emph{linear} weight allocation rule quickly performs better than strictly concave mappings

This is analytically seen most easily for the case in which all the involved distributions are normal. First, let all $\epsilon^l$ be distributed normally with zero mean and variance $\sigma_G^2$. Lemma~\ref{lem:Arnold_et_al} then implies that the median of $\{\epsilon^l\}_{l\in \mathcal{C}_i}$ is approximately normally distributed with zero mean and variance
$\pi \sigma_G^2/({2n_i}).$
Second, let the constituency-specific preference component $\mu_i$ be normally distributed with zero mean and variance $\sigma_H^2$. Constituency $\mathcal{C}_i$'s aggregate ideal point $\l_i$~-- the sum of two independent (approximately) normally distributed random variables~-- then also has an approximately normal distribution. Namely, \begin{equation}
\lambda_i \sim \mathbf{N}\left(0, \frac{\pi \sigma_G^2}{2n_i} + \sigma_H^2\right).
\end{equation}
Considering the corresponding densities at $\median=0$ for two representatives $i$ and $j$ yields
\begin{equation}
\frac{f_i(0)}{f_j(0)}=\left(\frac{\frac{\pi \sigma_G^2}{2n_i} + \sigma_H^2}{\frac{\pi \sigma_G^2}{2n_j} + \sigma_H^2}\right)^{-\frac{1}{2}}. 
\end{equation}
This ratio quickly approaches 1 as $\sigma^2_H \to \infty$, or if $\sigma^2_H>0$ and $n_i,n_j\to \infty$. Corollary~\ref{cor:general_solution_simple} then calls for $(w_1,\ldots, w_m)\propto (n_1, \ldots, n_m)$.

We pointed out in Section~\ref{sec:LimitAnalysis} that heterogeneity within each constituency will always give rise to different distributions of the sample medians when $n_i\neq n_j$. But the differences become small and no longer matter for pivotality in $\mathcal{R}^m$ when the heterogeneity across constituencies is sufficiently great.
This is illustrated by Figure~\ref{fig:density_illustration}. It depicts the density functions of ideal points $\l_i$ and $\l_j$ when $\mathcal{C}_i$ is four times larger than constituency $\mathcal{C}_j$, so that the standard deviation $\sigma_i$ of the median of $\{\epsilon^l\}_{l\in \mathcal{C}_i}$ is half the standard deviation $\sigma_j=\sigma$ of the median of $\{\epsilon^l\}_{l\in \mathcal{C}_j}$.
Panel~(a) shows the densities when $\sigma_H^2=0$ (or when we condition on $\mu_i=\mu_j=0$); panel~(b) depicts the case when $\mu_i, \mu_j\sim \mathbf{U}[-6\sigma,6\sigma]$. The densities of $\l_i$ and $\l_j$ in panel~(b) are very hard to distinguish in a neighborhood of the median $\median=0$.
This neighborhood's size increases in $\sigma_H^2$, and it coincides with the relevant policy range in which the Condorcet winner of $\mathcal{R}^m$ is most likely located under simple majority rule.

\begin{figure}[t]
{\small (a) \hspace{0.495\textwidth}(b) }\\
  \includegraphics[width=.495\textwidth]{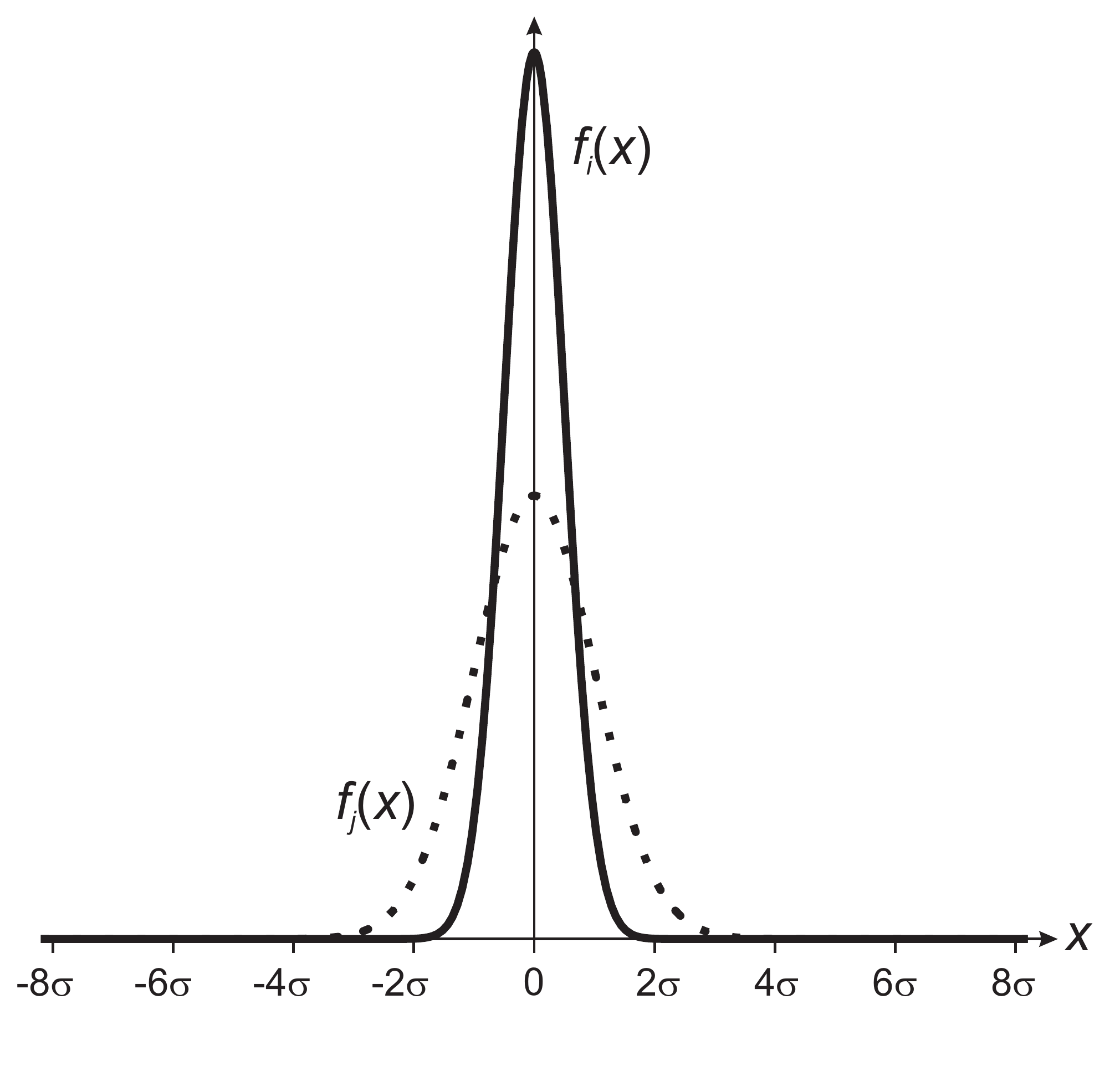}
  \includegraphics[width=.495\textwidth]{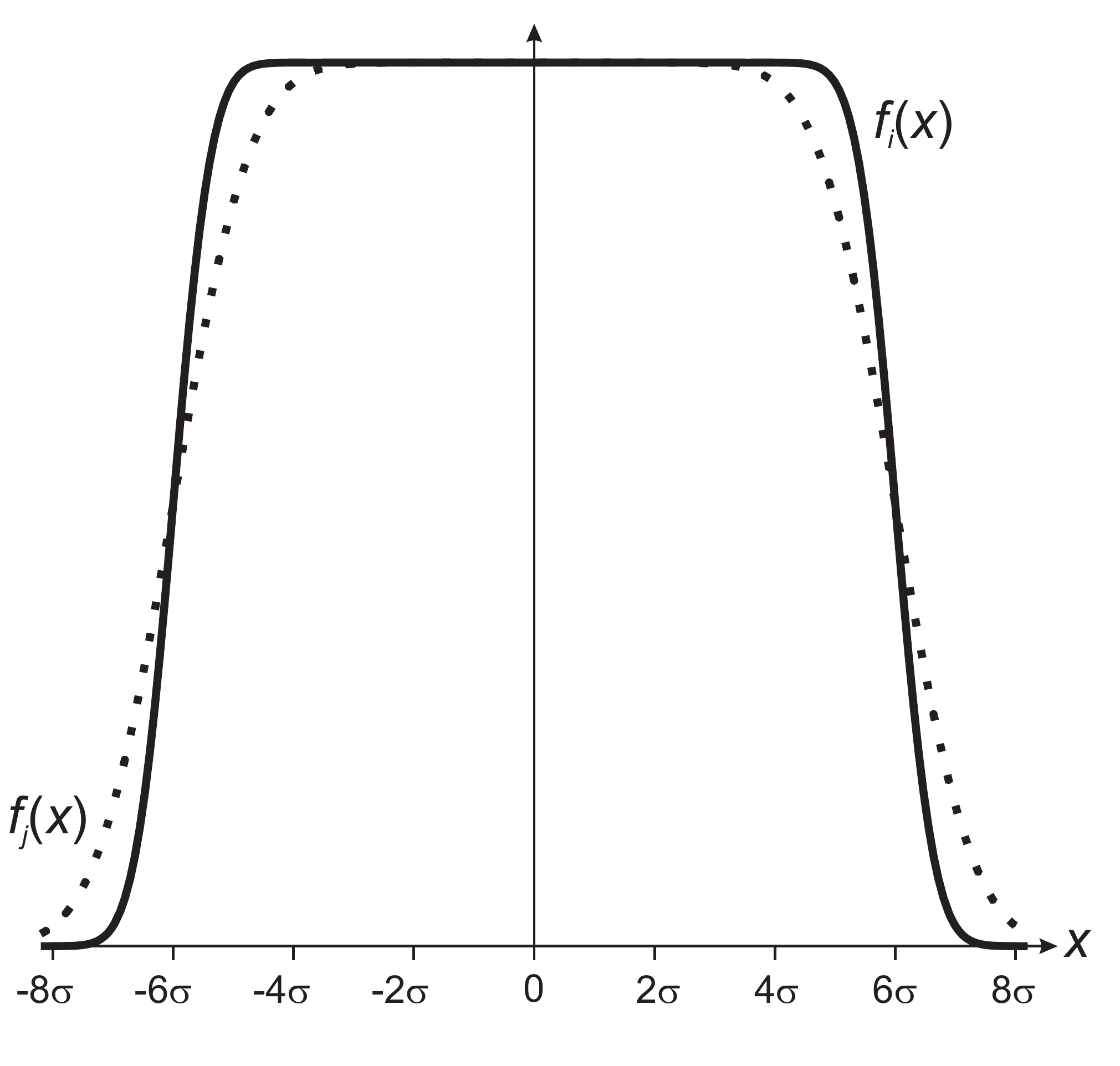}

  \caption{\small
  Densities of $\l_i$ and $\l_j$ when $n_i=4n_j$ and (a) $\mu_i=\mu_j= 0$ or (b) $\mu_i, \mu_j\sim \mathbf{U}[-6\sigma,6\sigma]$}
  \label{fig:density_illustration}
\end{figure}

Recall that the uniform distribution on $[a,b]$ has a variance of $(b-a)^2/12$. So panel~(b) shows a situation with $\sigma_H^2=12\sigma^2$. If we assume, as above, that all $\epsilon^l$ are normal with variance $\sigma_G^2$ then $\sigma_j=\sigma$ corresponds to
$\sigma_G^2=({2n_j}/{\pi}) \cdot \sigma^2$. Panel~(b) hence reflects a preference dissimilarity or polarization ratio of $\sigma^2_H/\sigma^2_G=6\pi/n_j$, which is tiny when one thinks of typical real-world population figures $n_j$.

\begin{figure}[t]
\vspace{\baselineskip}
\centering
  \includegraphics[width=0.875\textwidth]{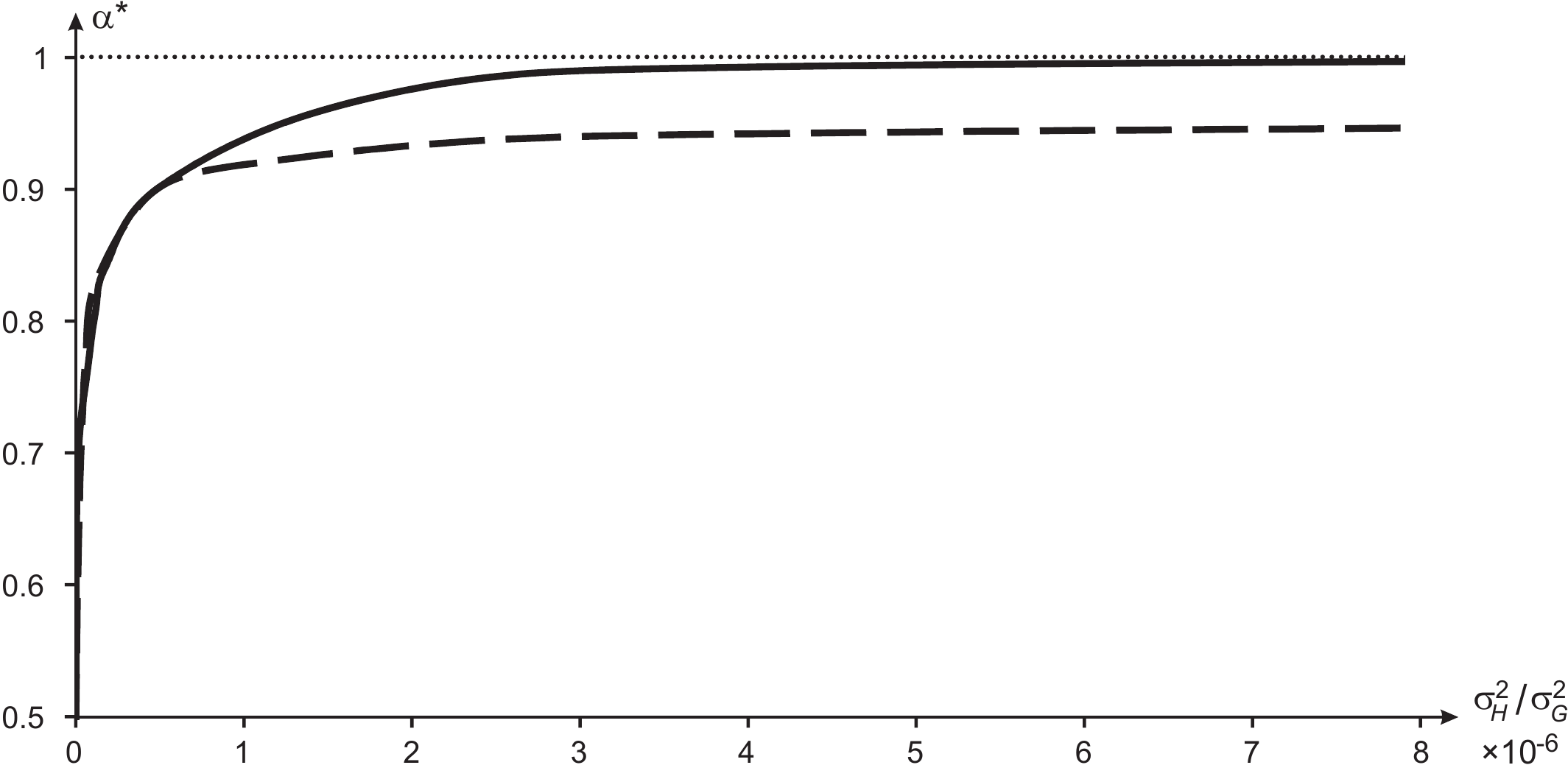} %
  \caption{\small
  Best coefficient $\alpha$ for direct (dashed line) and Shapley value-based allocation rules (solid line) with $n_1, \ldots, n_{27}$ defined by EU27 population data
}
  \label{fig:TransitionEU27}
\end{figure}

This suggests that the phase transition between optimality of a square root rule to optimality of a linear rule can be very fast. Figure~\ref{fig:TransitionEU27} demonstrates this when a population partition corresponding to the current European Union with 27 member states (EU27) is considered. The dashed line illustrates the (interpolated) optimal coefficients $\alpha^*$ as a function of $\sigma_H^2/\sigma_G^2$ when we search for the best rule in the class
\begin{equation}\label{eq:weight_based_alpha_rule}
(w_1, \ldots, w_m)\propto \left({n_1}^\alpha,\ldots, {n_m}^\alpha\right)
\end{equation}
for $\alpha\in \{0, 0.01, \ldots, 1.99, 2\}$;\footnote{Specifically, we consider $\epsilon^l\sim\mathbf{U}[-0.5, 0.5]$ and $\mu_i\sim\mathbf{N}(0,\sigma_H^2)$ with $0\le\sigma_H^2\le 10^{-6}$ and determine estimates of the pivot probabilities $\pi_i(\mathcal{R}^{27})$ which are induced by a given value of $\alpha$ via Monte Carlo simulation. The considered objective is to minimize $\Vert \cdot \Vert_1$-distance between individual pivot probabilities and the egalitarian ideal of $(1/n, \ldots, 1/n)\in \realnumbers^n$.} the solid line analogously depicts $\alpha^*$ when one searches within the class of Shapley value-based rules
\begin{equation}\label{eq:SSI_based_alpha_rule}
(w_1, \ldots, w_m) \textnormal{\ \,such that\ \,} \phi(q^m;w_1,\ldots, w_m) \propto \left({n_1}^\alpha,\ldots, {n_m}^\alpha\right).
\end{equation}
Optimality of the square root rule can be seen to break down very quickly; already small degrees of preference dissimilarity across constituencies render a linear rule based on the Shapley value optimal.\footnote{That $\alpha^*$ fails to converge to 1 when the simpler weight-based rule in (\ref{eq:weight_based_alpha_rule}) is concerned attests to the combinatorial nature of weighted voting, which cannot be totally ignored even for $m=27$. } This makes it possible to base design recommendations on rather qualitative assessments of polarization, i.e., it is not necessary to obtain precise estimates of $\sigma_H^2/\sigma_G^2$ in applications.

Note that Figure~\ref{fig:TransitionEU27} considers real EU population data but counterfactually assumes Council decisions to be taken by a simple majority. However, Figure~\ref{fig:density_illustration} suggests that a majority threshold of $q=50\%$ may not be a critical condition for optimality of a linear Shapley rule, provided that $\sigma_H^2>0$. We will make this claim precise in the remainder of the section.

When we presume that assembly $\mathcal{R}^m$ uses the 50\%-majority threshold defined in equation~(\ref{eq:qm_defined}), the representative $P\: m$ defined by (\ref{eq:P_defined}) can be considered as the pivotal member of $\mathcal{R}^m$ without much qualification. We can generalize our model and consider arbitrary relative majority thresholds $q\in [0.5;1)$ if we are willing to accept a weaker notion of pivotality. The complication is that the set of policy options that are $q$-majority undominated is no longer generically unique when $q>0.5$; supermajority rules induce cores which typically consist of entire intervals. We can, nevertheless, generalize the quota definition in (\ref{eq:qm_defined}) to
\begin{equation}\label{eq:qm_definition_generalized}
q^{m}\equiv q \sum_{j=1}^m w_j,
\end{equation}
for $q\in [0.5;1)$ and consider the representative $P\: m$ defined by (\ref{eq:P_defined}) to be pivotal. This may be justified most easily by supposing that a legislative status quo $x^\circ\approx \infty$ exists and that formation of a winning coalition proceeds qualitatively in the same fashion as is sometimes assumed in order to motivate the Shapley value:
coalition formation starts with the most enthusiastic supporters of change on the left, iteratively includes representatives further to the right, and gives all bargaining power to the first~-- and least enthusiastic~-- member who brings about the required supermajority.\footnote{Justifications for attributing most or all influence in $\mathcal{R}^m$ to representative $P\: m$ in the supermajority case date back to \citeN{Black:1948:EA}.
Distance-dependent costs of policy reform, a strategic external agenda setter, or the need of assembly $\mathcal{R}^{m}$ to bargain with outsiders can motivate a focus on the core's extreme points.
Status quo $x^\circ$ might also vary randomly on $X$ such that it lies to the left or right of the core equiprobably (with $\pi_i(\mathcal{R}^{m})$ then being $i$'s pivot probability conditional on policy change). }

Denote an $m$-member assembly $\mathcal{R}^m$ which uses the relative decision quota $q\in [0.5;1)$ and chooses policy $x^*=\l_{P\:m}$ as defined by (\ref{eq:P_defined})--(\ref{eq:outcome_equals_pivotal_point}) and (\ref{eq:qm_definition_generalized}) by $\mathcal{R}^{m,q}$. If $q>0.5$, the corresponding pivot probabilities $\pi_i(\mathcal{R}^{m,q})$ and $\pi_j(\mathcal{R}^{m,q})$ of representatives $i$ and $j$ in general \emph{fail} to exhibit the limit behavior with respect to $m$ which is characterized in Theorem~\ref{thm:main_theorem}. So Corollaries~\ref{cor:general_solution_simple} and \ref{cor:equal_representation_iid} do not apply when $q>0.5$.\footnote{One can check numerically that when one considers rules $(w_1, \ldots, w_m)\propto \left({n_1}^\alpha,\ldots, {n_m}^\alpha\right)$, the optimal coefficient $\alpha^*(q)$ for the i.i.d.\ case, where $\alpha^*(0.5)=0.5$, increases \emph{non-linearly} in $q$.}

However, a second asymptotic relationship applies for $q=0.5$ as well as arbitrary $q\in (0.5; 1)$, for arbitrary \emph{fixed} $m$, and without need for any kind of replica structure.
Specifically, we can consider the situation in which given non-degenerate shock variables $\mu_1,\ldots, \mu_m$, whose common probability density $h$ reflects preference heterogeneity across constituencies, are scaled by a non-negative factor $t$. Individual ideal points are then given by
\begin{equation}\label{eq:nu_decomposed_with_t}
\nu^l=t\cdot \mu_i+\epsilon^l
\end{equation}
for $t \ge 0$. The corresponding ideal point of representative~$i$ from constituency $\mathcal{C}_i$ is
\begin{equation}\label{eq:lambda_decomposed_with_t}
\l_i=t\cdot \mu_i+\tilde \epsilon_i
\end{equation}
with
\begin{equation}\label{eq:tilde_epsilon_defined}
\tilde \epsilon_i = \text{median\,}\{\epsilon^l\colon l\in \mathcal{C}_i\}
\end{equation}
where we maintain the assumption that all $\mu_i$ and $\epsilon^l$ are mutually independent and respectively identically distributed for $i\in \{1,\ldots, m\}$ and $l\in \{1,\ldots, n\}$.

The i.i.d.\ case amounts to $t=0$; and considering a large parameter $t$ corresponds to investigating an electorate which is highly polarized along constituency lines.
If we denote the pivot probability of representative~$i$ by $\pi_i(\mathcal{R}^{m,q,t})$ and the Shapley value of the weighted voting game $v=[q^m; w_1, \ldots, w_m]$ with $q^m$ defined by (\ref{eq:qm_definition_generalized}) as $\phi(v)$, the following holds:

\begin{theorem}\label{thm:main_result_for_shock_case}
Consider an assembly $\mathcal{R}^{m,q}$ with an arbitrary number $m$ of constituencies and the relative decision quota $q\in [0.5; 1)$. For each $i\in \{1, \ldots, m\}$ let $\l_i=t\cdot \mu_i+\tilde \epsilon_i$, where
$\mu_1,\ldots, \mu_m$ and $\tilde \epsilon_1,\ldots,\tilde \epsilon_m$ are all mutually independent random variables,
$\tilde \epsilon_1,\ldots,\tilde \epsilon_m$ have finite means and variances, and
$\mu_1,\ldots, \mu_m$ have an identical bounded density.
Then
\begin{equation}\label{eq:shockthm}
\lim_{\sigext\to \infty}\frac{\pi_i(\mathcal{R}^{m,q,t})}{\pi_j(\mathcal{R}^{m,q,t})} = \frac{\phi_i(v)}{\phi_j(v)}.
\end{equation}
\end{theorem}

\noindent The proof is provided in Appendix~B and formalizes that the respective orderings of representatives which are induced by $\l_1, \ldots, \l_m$ and by $t\cdot \mu_1, \ldots, t\cdot \mu_m$ tend to coincide when $t$ is large.\footnote{The density-driven intuition for Theorem~\ref{thm:main_result_for_shock_case} which is suggested by Figure~\ref{fig:density_illustration}(b) can also be made precise: under the additional assumption that the density $h$ of the shock terms $\mu_i$ is Lipschitz continuous, the density functions of $\l_1,\ldots,\l_m$ converge uniformly to that of $t\cdot \mu_i$. A proof is available from the authors. }
The theorem does not presume that $\tilde \epsilon_i$ satisfies (\ref{eq:tilde_epsilon_defined}); the limit (\ref{eq:shockthm}) applies also if $\l_i$ is determined, e.g., by an oligarchy instead of the median voter of $\mathcal{C}_i$.
It is, moreover, worth noting that Theorem~\ref{thm:main_result_for_shock_case} does not impose any conditions like Theorem~\ref{thm:main_theorem} on densities
$g_1, \ldots, g_m$ or voting weights $w_1, \ldots, w_m$ in assembly $\mathcal{R}^{m,q}$.
The Shapley value $\phi(v)$ automatically
takes care of any combinatorial particularities associated with $w_1, \ldots, w_m$;
and the convolution with $t\cdot \mu_i$'s bounded density,
$\frac{1}{t} h\left(\frac{x}{t}\right)$, is sufficient to `regularize' any (even non-continuous) distribution $G_i$ of $\tilde \epsilon_i$.
Applying Theorem~\ref{thm:main_result_for_shock_case} to the specific context of two-tier voting, we can conclude:

\begin{corollary}[Linear Shapley rule]\label{cor:equal_representation_shock_case}
If the ideal points of voters are the sum of an individual component $\epsilon^l$ which is  i.i.d.\ for all $l\in\{1, \ldots, n\}$ and a constituency-specific component $\mu_i$ which is i.i.d.\ for all $i\in \{1,\ldots, m\}$, representative~$i$'s ideal point equals the median voter's ideal point in constituency $\mathcal{C}_i$ for all $i\in \{1, \ldots, m\}$, and $\mu_i$'s variance is sufficiently great relative to that of $\epsilon^l$ then
\begin{equation}
(w_1, \ldots, w_m) \textnormal{\ \,such that\ \,} \phi(q^m;w_1,\ldots, w_m) \propto \left({n_1},\ldots, {n_m}\right)
\label{eq:shock_case_solution}
\end{equation}
achieves approximately equal representation for any given relative decision quota $q\in [0.5;1)$. 
\end{corollary}

The indirect representation of bottom-tier voters which is achieved by this linear Shapley rule can fail to be reasonably egalitarian
when $m$ is small, the distribution of constituency sizes is extremely skewed or has small variance, or when $q$ is close to 1.
This is because the so-called \emph{inverse problem} of finding weights which induce the desired Shapley value often fails to have a good solution in these cases.\footnote{This is easily seen, e.g., by considering constituencies of different sizes $n_1, \ldots, n_m$ and a relative quota $q\approx 1$ which essentially imposes unanimity rule; or by considering just $m=3$ constituencies, so that the only feasible Shapley values are~-- up to isomorphisms~-- $(1/3, 1/3, 1/3)$, $(2/3, 1/6, 1/6)$ and $(1, 0, 0)$. A new approach to solving the inverse problem exactly by using \emph{integer linear programming} has been proposed by \citeN{Kurz:2012}.}
Still, provided that the considered heterogeneity across constituencies is sufficiently bigger than the heterogeneity within, the indirect representation achieved by (\ref{eq:shock_case_solution}) is as egalitarian as possible.

Whether Corollary~\ref{cor:equal_representation_shock_case} for the case of noticeable preference affiliation within constituencies or Corollarly~\ref{cor:equal_representation_iid} for the i.i.d.\ case provides better guidance for designing a fair two-tier voting system in practice is hard to say.
Some preference homogeneity within and dissimilarity across constituencies seems plausible~-- whether as the result of a sorting process (`voting with one's feet') \`{a} la \citeN{Tiebout:1956}, due to cultural uniformity fostered by geographical proximity and local policies (see  \citeNP{Alesina/Spolaore:2003}), or for other reasons.
If constituencies correspond to entire nations, as in case of the EU~Council or ECB~Governing Council, citizens of a given constituency typically share more historical experience, traditions, language, communication etc.\ within constituencies than across.
(This plausibly is the key practical reason for why the issue of population size differences cannot trivially be resolved by redistricting in the first place.)
However, the collective decisions that are taken by the top-tier assembly might be primarily about issues where opinions range over the same liberal--conservative, markets--government, dove--hawk, etc.\ spectrum in all constituencies. Moreover, there might be normative reasons outside the scope of our analysis for \emph{pretending} that $\sigma_H^2=0$ even if it is not when one designs a presumably long-lasting, fair constitution.
We, therefore, avoid any specific recommendations here for, say, new voting rules in the EU~Council but warn that the i.i.d.\ presumption is more knife-edge and, therefore, seems to require special motivation.\footnote{A third alternative, inspired by the suggestion of ``flexible'' democratic mechanisms in other contexts (see Gersbach \citeyearNP{Gersbach:2005}, \citeyearNP{Gersbach:2009}), would be to specify different weighted voting rules for different policy domains. In some policy areas, such as competition policy, small or unstable between-constituency differences may call for square root weights; while fair decision making in other policy domains, such as agriculture or fisheries~-- with heterogenous shares of farmable land and some members landlocked, others islands~-- could involve linear weights.}

\section{Concluding remarks}\label{sec:Conclusion}

This paper has developed two limit results for the probability of being a decisive voter in order to address the issue of egalitarian representation of individuals in a two-tier voting system, such as the EU Council or the US Electoral College. Our concern was the equalization of the \emph{indirect influence} which bottom-tier voters can be expected to have on the collective decision in case of a one-dimensional convex policy space. The square root rule has played a prominent role in the related political discussion in the EU as well as the scientific discussion of \emph{binary} policy environments.
It was suggested to apply also more generally by the simulations of \citeN{Maaser/Napel:2007}.

We now provide it with a sound analytical foundation in a \emph{median voter} environment (Corollary~\ref{cor:equal_representation_iid}).
However, the somewhat counterintuitive square root rule turns out to have limited robustness. It does not extend to supermajority rules; it does badly in case of positive correlation of the ideal points at the constituency level. A linear rule quickly performs better and becomes optimal for sufficiently strong similarity within constituencies.

This dichotomy is, in some sense, not very surprising. The extensive literature on optimal voting weight allocations for binary policy alternatives has, for various objective functions, brought about either a square root or a linear rule (with few exceptions). Square root rules typically follow from far-reaching homogeneity and independence assumptions, while a linear rule is called for in case of dependence and significant across-constituency heterogeneity. For instance, \citeN{Kirsch:2007} finds square root weights to minimize the extent of disagreement between the council's binary decision and the popular vote with independent `yes' or `no' votes, but a linear rule if a sufficiently strong ``collective bias'' of the voters within each constituency is introduced.
The utilitarian design objective of \citeN{Barbera/Jackson:2006} calls for square root weights in their ``fixed-size-block model'', while they derive a linear rule in a ``fixed-number-of-blocks model'' which divides each constituency into the same number of blocks of identical voters.\footnote{The fixed-size-block model conceives of constituencies as consisting of many equally sized blocks of individuals whose preferences are perfectly correlated within a block and independent across blocks. The existence of such blocks~-- like those in the fixed-number-of-blocks model~-- would in our setup imply that, generically, no individual voter is ever pivotal in his constituency. Still, Theorems~\ref{thm:main_theorem} and \ref{thm:main_result_for_shock_case} could be used to characterize pivot probabilities of the respective representatives.}
\citeN{Beisbart/Bovens:2007} come to a very similar conclusion when trying to maximize welfare in another binary model: with i.i.d.\ utility parameters and simple majority rule, square root weights maximize total expected utility (and equalize it across citizens). But if an individual's utility is perfectly correlated with more other individuals the larger their constituency, then the square root rule quickly makes way for a proportional one.

So, using a very different and flexible framework, the corollaries  derived from two new limit results for interval policy spaces echo a pattern that has emerged also in the literature on binary two-tier voting systems. As originally argued by \citeN{Penrose:1946}, ex~ante independent and identical voters call for a voting weight allocation rule based on the square root of population sizes. However, sufficiently strong dissimilarity between constituencies renders most people's basic intuition correct~-- plain proportionality does the trick.

\newpage
\appendix
\appendixpage
\section{Proof of Theorem~\ref{thm:main_theorem}}\label{app:proof_main_theorem}

\begin{theoremiid*}
Consider an infinite chain $\mathcal{R}^{1}\subset \mathcal{R}^{2} \subset \mathcal{R}^{3}\subset \ldots$ of assemblies which involves a finite number $r$ of representative types, i.e., there exists a mapping $\tau\colon \naturalnumbers \to \{1, \ldots, r\}$ such that $\tau(j)=\t$ implies that $\l_j$ has density $f_\t$ and $w_j=w_\t\ge 0$.
Let the share of each type be bounded away from zero, i.e.,  there exist $\beta>0$ and $m^0\in \naturalnumbers$ such that $\beta_\t(m)\equiv |\mspace{1.5mu}\{k\in \{1, \ldots, m\}\colon \tau(k)=\t\}\mspace{2mu}|\mspace{2mu}/m \geq \beta>0$ for all $m\geq m^0$.
If for each $\t\in \{1, \ldots, r\}$ the distribution $F_\t$ has median $\median$ and its density $f_\t$ satisfies $f_\t(\median)>0$ with $|f_\t(x)-f_\t(\median)|\leq cx^2$ on a non-empty interval $[\median-\eps_1,\median+\eps_1]$ for some $c\geq 0$ then for $w_j>0$
\begin{equation*}
\lim_{m\to \infty}\frac{\pi_i(\mathcal{R}^{m})}{\pi_j(\mathcal{R}^{m})} = \frac{w_i f_{i}(\median)}{w_j f_{j}(\median)}.\tag{\ref{eq:thm}}
\end{equation*}
\end{theoremiid*}

\subsection{Overview}

Let us first give an overview of the five steps of the proof. In \emph{Step~1}, we define a particular neighborhood $I_m$ of the expected location of the weighted median of $\l_1, \ldots, \l_m$. This \emph{essential interval} shrinks to $\{\median\}$ as $m\to \infty$. It is constructed such that the probabilities $p_\t$, $\ptIl$, and $\ptIr$ of a type-$\t$ representative's ideal point falling inside $I_m$, inside $I_m$'s left half, or inside $I_m$'s right half, respectively, can suitably be bounded.
Moreover, we decompose the deterministic total number $m_\t=\beta_\t(m)\cdot m$ of type-$\t$ representatives in assembly $\mathcal{R}^m$ into the random numbers $\ntlinks$, $\ntI$, and $\ntrechts$ of delegates with ideal points to $I_m$'s left, inside $I_m$, and to $I_m$'s right.
Knowing the respective vector $\nvector=(\nlinks{1},\nmitte{1},\nrechts{1},\dots,\nlinks{r},\nmitte{r},\nrechts{r})$  will be sufficient to determine whether the Condorcet winner is located inside  $I_m$ or not.

In \emph{Step~2}, it is established that the weighted median of $\l_1, \ldots, \l_m$ is located inside the essential interval $I_m$ with a probability that quickly approaches~1 as $m\to \infty$.
As a corollary, the probability $\pi^\t(\mathcal{R}^m)$ of the Condorcet winner having type~$\t$ converges to the corresponding \emph{conditional} probability $\pi^\t(\mathcal{R}^m | \mathcal{K})$ of a type-$\t$ representative being pivotal where event $\mathcal{K}$ comprises all realizations of $\nvector$ such that $\mathcal{R}^m$'s weighted median lies inside $I_m$.

In \emph{Step~3}, we show that the random orderings of the $k=\sum_{\t\in\{1, \ldots, r\}} k_\t$ representatives with ideal point realizations $\l_i\in I_m$ asymptotically become equiprobable as $m\to \infty$. It follows that, with a vanishing error, the respective conditional pivot probability $\pi^\t(\mathcal{R}^m | \mathcal{K})$ equals the expected aggregate Shapley value of type-$\t$ representatives in $I_m$.

In \emph{Step~4}, the strong convergence result for the Shapley value by \citeN{Neyman:1982} is applied to our setting. Neyman's result implies that the aggregate Shapley value of type-$\t$ representatives with ideal points in $I_m$ converges to their respective aggregate voting weight in each considered weighted voting `subgame' among the representatives with ideal points $\l_i\in I_m$.

Having established that $\pi^\t(\mathcal{R}^m)$ is asymptotically proportional to the aggregate voting weight of all type-$\t$ representatives with ideal points inside $I_m$, aggregate probabilities are attributed to individual representatives in the final \emph{Step~5}.

\subsection{Proof}
\subsubsection*{Step 1: Essential interval $\boldsymbol{I_m}$ and vector $\nvector$}
We begin by identifying a neighborhood of $\median$ and a sufficiently great number of representatives such that both the densities $f_\t$ and the numbers of type-$\t$ representatives in $\mathcal{R}^m$ can suitably be bounded. This leads to the definition of intervals $I_m$ around $\median$ which later steps will focus on. Bounds for the probabilities of a type-$\t$ representative's ideal point falling inside $I_m$, and more specifically into $I_m$'s left or right halves, are provided in Lemma~\ref{lem:ptheta_bounds}. The final part of \emph{Step~1} introduces the vector $\nvector$ as a type-specific summary of how many ideal points are located to the left of $I_m$, inside $I_m$, and to its right.

First note that
\begin{align}\label{eq:inequ1}
0<\uline{u} \equiv \min_{\t'\in \{1,\ldots,r\}} f_{\t'}(\median) \leq {f}_{\t}(\median)\leq \overline{u}\equiv \max_{\t'\in \{1,\ldots,r\}} f_{\t'}(\median)
\end{align}
for every $\t\in \{1, \ldots, r\}$. Using the continuity of $f_\t$ in a neighborhood of $\median$, which is implied by $|f_\t(x)-f_\t(\median)|\leq cx^2$, we can choose $0<\eps_2\leq \eps_1$ such that
\begin{align}\label{ie:bounds_f}
\frac{5}{6}{f}_{\t}(\median) \leq {f}_{\t}(x)\leq \frac{7}{6}{f}_{\t}(\median)
\end{align}
for all $x\in [\median-\eps_2,\median+\eps_2]$ and any specific $\t\in \{1, \ldots, r\}$.
Inequality (\ref{eq:inequ1}) can be used in order to obtain bounds
\begin{align}\label{ie:bounds_u}
\frac{1}{2}\uline{u} \leq {f}_{\t}(x)\leq 2\overline{u}
\end{align}
for all $x\in [\median-\eps_2,\median+\eps_2]$ and all $\t\in \{1, \ldots, r\}$ which do not depend on $\t$. Due to the
existence of $m^0$ we can also choose $0<\eps_3\leq \eps_2$ such that
\begin{align}\label{ie:bounds_beta}
\beta_\t(m)\ge \beta>0
\end{align}
for all $m\geq \frac{1}{{\eps_3}^{8/3}}$ and all $\t\in \{1, \ldots, r\}$. And we can determine $0<\eps_4\leq \eps_3$ such that
\begin{align}\label{ie:technical_bound_constants}
24<\uline{u}\beta\cdot (m \beta)^{\frac{1}{40}}\le \uline{u}\beta m_\t^{\frac{1}{40}}
\end{align}
for all $m\geq \frac{1}{{\eps_4}^{8/3}}$, where $m_\t\equiv\beta_\t(m)\cdot m$.

Then define
\begin{align}\label{eq:epsmdef}
\eps(m)\equiv m^{-\frac{3}{8}}
\end{align}
and note that $\eps(m)\leq \eps_4$ iff $m\geq m^1\equiv \frac{1}{{\eps_4}^{8/3}}\ge m^0$. So, whenever we consider a sufficiently large number of representatives (specifically, $m\geq m^1$), 
inequalities (\ref{ie:bounds_f})--(\ref{ie:technical_bound_constants}) are satisfied.
We refer to
\begin{equation}
I_m\equiv [\median-\eps(m), \median+\eps(m)]
\end{equation}
as the \emph{essential interval}.
The probability of an ideal point of type $\t$ to fall inside $I_m$ is
\begin{align}
\ptI\equiv\int\limits_{\median-\varepsilon(m)}^{\median+\varepsilon(m)} f_\t(x)dx.
\end{align}
For realizations in the left and right halves of $I_m$ we respectively obtain the probabilities
\begin{equation}\label{eq:ptheta_left_right_defined}
\ptIl\equiv\int\limits_{\median-\varepsilon(m)}^{\median} f_\t(x)dx
\quad \text{and}\quad \ptIr\equiv\int\limits_{\median}^{\median+\varepsilon(m)} f_\t(x)dx,
\end{equation}
with $\ptIl+\ptIr=\ptI$.

\begin{lemma}\label{lem:ptheta_bounds}
For $m\geq m^1$ we have
\begin{eqnarray}
\frac{5}{3}f_\t(\median) \varepsilon(m)\le & \ptI & \le \frac{7}{3}f_\t(\median) \varepsilon(m), \label{eq:ptheta_bound_f}\\
\frac{5}{6}f_\t(\median) \varepsilon(m) \le & \ptIl,\, \ptIr & \le \frac{7}{6}f_\t(\median) \varepsilon(m), \label{eq:ptheta_left_right_bounds_A} \\
  \underline{u}\beta m_\t^{-\frac{3}{8}} \le 
  & \ptI & \le
4\overline{u}m_\t^{-\frac{3}{8}}, \label{eq:ptheta_bound} \text{\ \ and}\\
\frac{1}{2} \underline{u}\beta m_\t^{-\frac{3}{8}}\le  & \ptIl,\, \ptIr & \le 2\overline{u}m_\t^{-\frac{3}{8}}. \label{eq:ptheta_left_right_bounds_B}
\end{eqnarray}
\end{lemma}
\begin{proof}
The inequalities can be concluded from (\ref{ie:bounds_f})--(\ref{ie:bounds_beta}), $m_\t=\beta_\t m$, and $\beta<1$.
\end{proof}

Now for any realization $\boldsymbol{\l}$ of the ideal points in assembly $\mathcal{R}^{m}$, let
\begin{align}
  \ntI \equiv\#\{j\colon \tau(j)=\t \mbox{ and } \l_j\in [\median-\eps(m), \median+\eps(m)]\}
\end{align}
denote the number of type-$\t$ representatives with a policy position in the essential interval, i.e., no more than $\eps(m)$ away from the expected sample median $\median$. Analogously, let
\begin{align}
  \ntlinks  \equiv\#\{j\colon \tau(j)=\t \mbox{ and } \l_j\in (-\infty, \median-\eps(m))\}
\end{align}
and
\begin{align}
  \ntrechts  \equiv\#\{j\colon \tau(j)=\t \mbox{ and } \l_j\in (\median+\eps(m), \infty) \}
\end{align}
denote the random number of type-$\t$ representatives to the left and to the right of $I_m$.

One can conceive of $\boldsymbol{\l}$-realizations as the results of a two-part random experiment: in the first part, it is determined for
each $\l_j$ whether it is located to the right of $I_m$, to its left, or inside $I_m$, e.g., by drawing a vector $\boldsymbol{l}=(l_1, \ldots, l_m)$ of
independent random variables where $l_i=1$ $(-1)$ indicates a realization of $\l_i$ to the right (left) of $I_m$ and $l_i=0$ indicates $\l_i\in I_m$ (with probabilities  $\frac{1}{2}-\ptIl$, $\frac{1}{2}-\ptIr$, and $\ptI$,
respectively).
This already fixes $\ntlinks$, $\ntI$, and $\ntrechts$ for each $\t\in \{1, \ldots, r\}$ and is summarized by the vector
\begin{equation}
\nvector=(\nlinks{1},\nmitte{1},\nrechts{1},\dots,\nlinks{r},\nmitte{r},\nrechts{r}).
\end{equation}
In the second part, the exact ideal point locations are drawn. It will turn out that those outside $I_m$ can be ignored with vanishing error; and
the $k_\t$ type-$\t$ ideal points inside have conditional densities $\hat {f_\t}$ with
\begin{align}\label{eq_conditional_densitiy}
\hat {f_\t}(x)\equiv \frac{f_{\t}(x)}{\ptI} \quad \mbox{for } x\in I_m.
\end{align}

\subsubsection*{Step 2: Type $\boldsymbol{\t}$'s aggregate pivot probability $\boldsymbol{\pi^\t(\mathcal{R}^{m})}$ converges to the conditional probability $\boldsymbol{\pi^\t(\mathcal{R}^{m}|\mathcal{K})}$ of type $\boldsymbol{\t}$ being pivotal in $\boldsymbol{I_m}$}
We next appeal to \emph{Hoeffding's inequality}\footnote{See \citeN[Theorem~2]{Hoeffding:1963}.} in order to obtain bounds on the probability that the shares of representatives $\frac{\ntlinks}{m_\t}$, $\frac{\ntI}{m_\t}$, and $\frac{\ntrechts}{m_\t}$ with ideal points to the left, inside, or right of $I_m$ deviate by more than a specified distance from their expectations.
These bounds will imply that one can condition on the pivotal ideal point lying inside $I_m$ in later steps of the proof with an exponentially decreasing error.

Hoeffding's inequality concerns the average $\overline{X}\equiv \frac{1}{n}\cdot\sum\limits_{i=1}^n X_i$ of $n$ independent bounded random variables $X_i\in[a_i,b_i]$ and guarantees
\begin{align}\label{eq:Hoeffding}
  \Pr\left\{\left|\overline{X}-\E[\overline{X}]\right|>t\right\}\le 2\exp\left(\frac{-2t^2n^2}{\sum\limits_{i=1}^n(b_i-a_i)^2}\right).
\end{align}
Our specific construction will involve only random variables $X_i\in[0,1]$, so that
\begin{align}\label{eq:Hoeffding01}
  \Pr\left\{\left|\overline{X}-\E[\overline{X}]\right|>t\right\}\le 2\exp\left(-2t^2n\right).
\end{align}
We will put $n=m_\t$ for a fixed $\t\in \{1, \ldots, r\}$, so that $n\to\infty$ as $m\to \infty$, and choose $t=n^{-\frac{2}{5}}$, which implies $t(n)\ll \eps(m)$. For this choice
\begin{align}\label{eq:Hoeffding_for_us}
  \Pr\left\{\left|\overline{X}-\E[\overline{X}]\right|>n^{-\frac{2}{5}}\right\}\le 2\exp\left(-2n^{\frac{1}{5}}\right),
\end{align}
i.e., the probability of ``extreme realizations'' exponentially goes to zero as $m\to \infty$ (and hence $n=m_\t \to \infty$).

\begin{lemma}
\label{lemma:concentration_per_type}
For each $\t\in\{1,\dots, r\}$ we have:
\begin{equation*}
\renewcommand{\arraystretch}{2.5}
\begin{array}{llcl}
\textnormal{(I)} &
\Pr\left\{\frac{\ntlinks}{m_\t}\in \left[
\frac{1}{2}-\ptIl
- m_\t^{-\frac{2}{5}},
\frac{1}{2}-\ptIl
+ m_\t^{-\frac{2}{5}}\right]\right\}
&\geq& 1-2\exp\left(-2m_\t^{\frac{1}{5}}\right) \\
\textnormal{(II)} &
\Pr\left\{\frac{\ntI}{m_\t}\in \left[\ptI- m_\t^{-\frac{2}{5}}, \ptI+ m_\t^{-\frac{2}{5}} \right] \right\}
&\geq& 1-2\exp\left(-2m_\t^{\frac{1}{5}}\right) \\
\textnormal{(III)} &
\Pr\left\{\frac{\ntrechts}{m_\t}\in \left[
\frac{1}{2}-\ptIr
- m_\t^{-\frac{2}{5}},
\frac{1}{2}-\ptIr
+ m_\t^{-\frac{2}{5}}\right]\right\}
&\geq& 1-2\exp\left(-2m_\t^{\frac{1}{5}}\right).
\end{array}
\renewcommand{\arraystretch}{1}
\end{equation*}
\end{lemma}
\begin{proof}
Let $\t\in\{1,\dots,r\}$ be arbitrary but fixed. For statement~(I) we consider the $n=m_\t$ indices $j_1\dots,j_{m_\t}\in\{1,\dots, m\}$ of type $\t$
and denote by $X_i$ the random variable which is $1$ if the realization $\lambda_{j_i}$ lies inside the interval $(-\infty,\median-\varepsilon(m))$ and zero otherwise. In the
notation of Hoeffding's inequality we have $\overline{X}=\frac{\ntlinks}{m_\t}$.
Since the probability that $\lambda_{j_i}$ lies in the left half of $I_m$ is given by $\ptIl$ and
$\int_{-\infty}^{\median} {f}_\t(x)dx=\int_{\median}^{\infty}{f}_\t(x)dx=\frac{1}{2}$, the probability that $\lambda_{j_i}$ lies in the
interval $(-\infty,\median-\varepsilon(m))$ is given by
$\frac{1}{2}-\ptIl$. Thus we have $\E[\overline{X}]=\frac{1}{2}-\ptIl$ and (\ref{eq:Hoeffding_for_us}) implies (I).
The statements (II) and (III) follow along the same lines (namely, by letting $X_i$ be the characteristic function of intervals
$[\median-\varepsilon(m),\median+\varepsilon(m)]$ and
$(\median+\varepsilon(m),\infty)$, respectively). Note that ${m_\t}^{-2/5}\ll \eps(m)=m^{-3/8}$
for large $m$.
\end{proof}

We can use the bounds on $\ptI$ in (\ref{eq:ptheta_bound}) and that $\beta m\le m_\t \le m$ for $m\ge m^1\ge m^0$ in order to conclude from (II) that for any given $\t\in \{1,\ldots,r\}$
\begin{align}\label{eq:bounds_on_ktheta}
\underline{u}\beta^2\varepsilon(m)\cdot m-m^{\frac{3}{5}}\le \ntI\le 4\overline{u}\varepsilon(m)\cdot m +m^{\frac{3}{5}}
\end{align}
with a probability of at least $1-2\cdot \exp\left(-2{m_\t}^{\frac{1}{5}}\right)$.
A further implication of observations (I)--(III) is:
\begin{lemma}
  \label{lemma:unweighted_median_per_type}
For $m\ge m^1$ the inequalities
\begin{eqnarray}
\ntlinks &<& \frac{1}{2}m_\t  \\
\ntrechts &<& \frac{1}{2}m_\t \\
\ntlinks+\frac{2}{3}\ntI &>& \frac{1}{2}m_\t \\
\ntrechts+\frac{2}{3}\ntI &>& \frac{1}{2}m_\t
\end{eqnarray}
are simultaneously satisfied for \emph{all} $\t\in\{1,\dots,r\}$ with a probability of 
at least $1-6r\cdot \exp\left(-2{(\beta m)}^{\frac{1}{5}}\right)$.
\end{lemma}

\vspace{0.2cm}

\begin{proof}
The events considered in statements (I), (II), and (III) of Lemma~\ref{lemma:concentration_per_type} are realized for \emph{all} $\t\in\{1,\dots,r\}$ with a joint probability of at least
\begin{align}
\left(1-2\exp\left(-2{(\beta m)}^{\frac{1}{5}}\right)\right)^{3r}\ge 1-6r\exp\left(-2{(\beta m)}^{\frac{1}{5}}\right),
\end{align}
since $m_\t\ge \beta m$ for $m\ge m^0$ and $(1-x)^k\ge (1-kx)$ is valid for all $x\in[0,1]$ and $k\in\mathbb{N}$. 
If $m\geq m^1$, we then have
\begin{align}
     \ntlinks \le \left(\frac{1}{2}-\ptIl\right)\cdot m_\t+{m_\t}^{\frac{3}{5}} \le \frac{m_\t}{2}-\frac{\underline{u}
    \beta{m_\t}^{\frac{5}{8}}}{2} +{m_\t}^{\frac{3}{5}}=
    \frac{m_\t}{2}-{m_\t}^{\frac{3}{5}}\underset{>0} {\underbrace{\left(\frac{\underline{u}\beta m_\t^{\frac{1}{40}}}{2}-1\right)}}
    <\frac{1}{2}m_\t
\end{align}
for any $\t\in\{1,\dots,r\}$. The first inequality follows directly from (I), the second inequality uses
(\ref{eq:ptheta_left_right_bounds_B}), and the final inequality follows from (\ref{ie:technical_bound_constants}). Analogous inequalities pertain to $\ntrechts$.

Moreover, we can conclude
\begin{eqnarray}
\ntlinks+\frac{2}{3}\ntI &\ge &
\left(\frac{1}{2}-\ptIl\right)\cdot m_\t- {m_\t}^{\frac{3}{5}}+\frac{2\ptI}{3}m_\t-\frac{2}{3}{m_\t}^{\frac{3}{5}} \\
&=&
\frac{m_\t}{2} - \frac{5}{3}{m_\t}^{\frac{3}{5}} +
\left(\frac{2\ptI}{3}-\ptIl\right)m_\t
\\
&=&
\frac{m_\t}{2} + \frac{5}{3}{m_\t}^{\frac{3}{5}}\left(\frac{2\ptIr}{5}{m_\t}^{\frac{2}{5}}-\frac{\ptIl}{5}{m_\t}^{\frac{2}{5}} -1 \right)\\
&\ge&
\frac{m_\t}{2} + \frac{5}{3}{m_\t}^{\frac{3}{5}}\left(\frac{1}{10}\cdot\frac{3}{7}\ptI \cdot{m_\t}^{\frac{2}{5}}-1 \right) \\
&\ge &\frac{m_\t}{2} +    \frac{5}{3}{m_\t}^{\frac{3}{5}}\underset{>0} {\underbrace{\left(\frac{\underline{u}\beta{m_\t}^{\frac{1}{40}}}
{24}-1\right)}} >\frac{1}{2}m_\t.
\end{eqnarray}
The first inequality uses (I) and (II); the second one employs (\ref{eq:ptheta_bound_f}) and (\ref{eq:ptheta_left_right_bounds_A}); the third applies (\ref{eq:ptheta_bound});
and the final one invokes (\ref{ie:technical_bound_constants}). Analogous inequalities pertain to $\ntrechts+\frac{2}{3}\ntI$.
\end{proof}

Lemma~\ref{lemma:unweighted_median_per_type} implies that the respective unweighted sample median among representatives of type~$\t$ is located within $I_m$ for all $\t\in\{1, \ldots, r\}$ with a probability that quickly approaches 1. The same must \emph{a~fortiori} be true for the pivotal assembly member, i.e., the weighted median among all representatives.

We collect in the set $\mathcal{K}$ all $\nvector=(\nlinks{1},\nmitte{1},\nrechts{1},\dots,\nlinks{r},\nmitte{r},\nrechts{r})$ such that the events considered by Lemma~\ref{lemma:concentration_per_type}, (I)--(III), are realized for \emph{all} $\t\in \{1, \ldots, r\}$. The inequalities in Lemma~\ref{lemma:unweighted_median_per_type} then hold for any $\nvector \in \mathcal{K}$.
We can decompose the probability $\pi^\t(\mathcal{R}^{m})$ of some type-$\t$ representative being pivotal into conditional probabilities $\pi^\t(\mathcal{R}^{m}|\mathcal{K})$ and $\pi^\t(\mathcal{R}^{m}|\neg\mathcal{K})$ which respectively concern only $\boldsymbol{\l}$-realizations where $\nvector\in \mathcal{K}$ and $\nvector\not\in \mathcal{K}$. Then  Lemma~\ref{lemma:unweighted_median_per_type} implies
\begin{eqnarray}
\pi^\t(\mathcal{R}^{m})&=&\Pr\{\mathcal{K}\}\cdot \pi^\t(\mathcal{R}^{m}|\mathcal{K}) + \Pr\{\neg\mathcal{K}\}\cdot \pi^\t(\mathcal{R}^{m}|\neg \mathcal{K}) \notag\\
&=&\pi^\t(\mathcal{R}^{m}|\mathcal{K}) + O(exp(-2m^{\frac{1}{5}})). \label{eq:TypSSIgesplitted}
\end{eqnarray}

\subsubsection*{Step 3: $\boldsymbol{\pi^\t(\mathcal{R}^{m}|\mathcal{K})}$ converges to the expectation of type $\t$'s Shapley value inside $\boldsymbol{I_m}$}
Now condition on some $\nvector\in \mathcal{K}$ such that exactly $\sum_{\theta} \ntI=k$ ideal points fall inside the essential interval, where $k$ is asymptotically proportional to $\eps(m)\cdot m=m^{\frac{5}{8}}$ by (\ref{eq:bounds_on_ktheta}). 
Label them $1,\dots, k$ for ease of notation and let $\ordering\in\mathcal{S}_k$ denote an arbitrary element of the space $\mathcal{S}_k$ of permutations which bijectively map $(1,\dots,k)$ to some $(j_1,\dots,j_k)$.
The conditional probability for the event that the $k$ ideal points located in $I_m$ are ordered exactly as they are in $\ordering$ by the second step of the experiment is
\begin{equation}
  \pperm\equiv \int_{-\varepsilon(m)}^{\varepsilon(m)}\int_{x_{j_{1}}}^{\varepsilon(m)}\dots  \int_{x_{j_{k-1}}}^{\varepsilon(m)}
  \hat{f}_{j_1}(x_{j_1})\dots \hat{f}_{j_k}(x_{j_k})\,dx_{j_k}\dots dx_{j_{2}}dx_{j_1}.
\end{equation}

\begin{lemma} For all $m\geq m^1$, any $\nvector\in \mathcal{K}$ with $\sum_{\theta} \ntI=k$ and permutation $\ordering\in\mathcal{S}_k$ we have
  \label{lemma:permutation_equiprobability}
\begin{equation}
  \pperm=\frac{1}{k!} + \frac{1}{k!}\cdot O( m^{-\frac{1}{8}}) .
\end{equation}
\end{lemma}

\vspace{0.2cm}

\begin{proof}

The premise $|f_\t(x)-f_\t(\median)|\leq cx^2$ for $x\in I_m$ permits us to choose $\delta \in O(\eps(m)^2)$
with $\delta\le \frac{1}{2}$ such that
\begin{equation}\label{eq:f_t_bounding}
  (1-\delta)\cdot f_{\t}(\median) \le f_{\t}(x) \le (1+\delta)\cdot f_{\t}(\median)
\end{equation}
and, equivalently,
\begin{equation}\label{eq:f_t_hat_bounding}
  (1-\delta)\cdot \hat f_{\t}(\median) \le \hat f_{\t}(x) \le (1+\delta)\cdot \hat f_{\t}(\median)
\end{equation}
for all types $1\le \t\le r$ and all $x\in I_m$.
Integrating (\ref{eq:f_t_bounding}) on $I_m$ yields
\begin{equation}\label{eq:p_t_bounding}
  2\eps(m)(1-\delta)\cdot f_{\t}(\median) \le \ptI  \le 2\eps(m)(1+\delta)\cdot f_{\t}(\median).
\end{equation}
With these bounds we can conclude from $\hat {f_\t}(\median)=\frac{f_\t(\median)}{\ptI}$ that
\begin{equation}
  \label{ie:hat_constant}
  \frac{1-\delta}{2\eps(m)}\le\frac{1}{2\eps(m)(1+\delta)} \le \hat {f_\t}(\median) \le  \frac{1}{2\eps(m)(1-\delta)}
  \le \frac{1+2\delta}{2\eps(m)}
\end{equation}
because $1/(1-\d)\le 1+2\d$.

Using $(1-\delta)^k\ge 1-k\delta$ and
$(1+\delta)^k\le 1+2k\delta$ for $k\delta\le 1$,\footnote{The first statement is easily seen by induction on $k$. The second follows from
$$
(1+\delta)^k=\sum_{j=0}^k {k\choose j}\delta^j \le
1+\sum_{j=1}^k \frac{1}{j!} \underbrace{(k\delta)^j}_{\le k\delta}\le 1+k\delta\underbrace{\sum_{j=1}^k\textstyle \frac{1}{j!}}_{\le e-1}\le 1+2k\delta.
$$
Since $k$ is asymptotically proportional to $m^{\frac{5}{8}}$ and $\eps(m)^2=m^{-\frac{6}{8}}$ we can choose $\delta\in O(m^{-\frac{6}{8}})$ with $(k\d)^j\le k\d$ for $j\ge 1$ whenever $m$ is large enough.\label{fn:binomial_inequality}} and noting that the hypercube $[0,1]^k$ can be partitioned into $k!$ polytopes $\{x\in [0,1]^k\colon x_{j_1}\le x_{j_2} \le \ldots \le x_{j_k}\}$ with equal volume,
inequality (\ref{eq:f_t_hat_bounding}) yields
\begin{eqnarray}
  \pperm&\ge& (1-\delta)^k\int_{-\varepsilon(m)}^{\varepsilon(m)}\int_{x_{j_{1}}}^{\varepsilon(m)}\dots  \int_{x_{j_{k-1}}}^{\varepsilon(m)}
  \hat{f}_{j_1}(\median)\dots \hat{f}_{j_k}(\median)\,dx_{j_k}\dots dx_{j_{2}}dx_{j_1}\\
  &=& \frac{(1-\delta)^k}{k!}\cdot \hat{f}_{j_1}(\median)\dots \hat{f}_{j_k}(\median) \int_{-\varepsilon(m)}^{\varepsilon(m)}
  \int_{-\varepsilon(m)}^{\varepsilon(m)}\dots \int_{-\varepsilon(m)}^{\varepsilon(m)} 1
  \,dx_{j_k}\dots dx_{j_{2}}dx_{j_1}\\
 &=& \frac{(1-\delta)^k}{k!}\cdot \hat{f}_{j_1}(\median)\dots \hat{f}_{j_k}(\median) \cdot (2\eps(m))^k \\
 &\overset{(\ref{ie:hat_constant})}{\ge}& \frac{(1-\delta)^{2k}}{k!}\ge \frac{1-2k\delta}{k!}
\end{eqnarray}
and, analogously,
\begin{eqnarray}
  \pperm&\le& (1+\delta)^k\int_{-\varepsilon(m)}^{\varepsilon(m)}\int_{x_{j_{1}}}^{\varepsilon(m)}\dots  \int_{x_{j_{k-1}}}^{\varepsilon(m)}
  \hat{f}_{j_1}(\median)\dots \hat{f}_{j_k}(\median)\,dx_{j_k}\dots dx_{j_{2}}dx_{j_1}\\
 &=& \frac{(1+\delta)^k}{k!}\cdot \hat{f}_{j_1}(\median)\dots \hat{f}_{j_k}(\median) \cdot (2\eps(m))^k \\
 &\overset{(\ref{ie:hat_constant})}{\le}& \frac{(1+\delta)^{k}(1+2\delta)^k}{k!}
 \le \frac{(1+2\delta)^{2k}}{k!}\le \frac{1+8k\delta}{k!}.
\end{eqnarray}
This implies
\begin{equation}
\left|\pperm-\frac{1}{k!}\right|\le \frac{8k\delta}{k!}.
\end{equation}
Because $k\in O(m^\frac{5}{8})$ and $\delta \in O(m^{-\frac{6}{8}})$, the relative error $|\pperm-(k!)^{-1}|\big/ (k!)^{-1}$ tends to zero
at least as fast as $O(m^{-\frac{1}{8}})$.
\end{proof}

So even though the probabilities of the orderings $\ordering\in\mathcal{S}_k$ of the $k$ agents inside $I_m$ differ depending on which specific $\ordering$ is considered and what are the involved representative types (i.e., which $\nvector$ is considered), these differences vanish and all orderings become equiprobable as $m$ gets large.

Type $\t$'s conditional pivot probability can be written as
\begin{equation}\label{eq:cond_pivot_prob_per_type}
\pi^\t(\mathcal{R}^{m}|\mathcal{K})=
  \sum_{\nvector\in\mathcal{K}}\pvect \cdot \Big\{ \sum_{\ordering\in \mathcal{S}_{k}\,:\,\psi(\nvector,\ordering)=\t} \pperm \Big\},
\end{equation}
where $P(\nvector)$ denotes the probability of $\nvector$ conditional on event $\{\nvector\in \mathcal{K}\}$ and function $\psi\colon \mathcal{K}\times \mathcal{S}_k\to  \{1,\ldots,r\}$ identifies the type $\t'$ of the pivotal member in $\mathcal{R}^{m}$ when $\nvector$ describes how the representative types are divided between $I_m$ and its left or right, and $\ordering$ captures the ordering inside $I_m$.  Lemma~\ref{lemma:permutation_equiprobability} approximates the probability of ordering $\ordering$ conditional on $\nvector$ as $1/k!$, and one thus obtains
\begin{equation}\label{eq:cond_pivot_prob_as_expected_Shapley}
\pi^\t(\mathcal{R}^{m}|\mathcal{K})=
  \sum_{\nvector\in\mathcal{K}}\pvect \cdot
{\phi_\theta(\nvector)} + O(m^{-\frac{1}{8}})
\end{equation}
with
\begin{equation}
\phi_\theta(\nvector)=\sum_{\ordering\in \mathcal{S}_{k}\,:\,\psi(\nvector,\ordering)=\t} \frac{1}{k!}.
\end{equation}
Because a constant factor $\frac{1}{k!}$ pertains to each ordering $\ordering\in \mathcal{S}_k$, $\phi_\theta(\nvector)$ equals the probability that, as the weights $w_1, w_2, \ldots, w_k$ of the $k$ representatives inside $I_m$ are accumulated in
\emph{uniform} random order, the threshold $q(\nvector)\equiv q^{m}-\sum_{\t\in \{1, \ldots, r\}}\ntlinks w_\t$ is first reached by the weight of a type-$\t$ representative. The term $\phi_\theta(\nvector)$ is, therefore, simply the aggregated \emph{Shapley value} of the type-$\t$ representatives in the weighted voting game defined by quota $q(\nvector)$ and weight vector $(w_1, w_2, \ldots, w_k)$. Equation~(\ref{eq:cond_pivot_prob_as_expected_Shapley}) states that $\pi^\t(\mathcal{R}^{m}|\mathcal{K})$ converges to the expectation of this Shapley value $\phi_\theta(\nvector)$.

\subsubsection*{Step 4: Type $\boldsymbol{\t}$'s Shapley value $\boldsymbol{\phi_\theta(\nvector)}$ converges to $\boldsymbol{\t}$'s relative weight in $\boldsymbol{I_m}$}

Condition $\nvector\in \mathcal{K}$ implies $\frac{1}{3}\cdot \sum_{\t \in \{1, \ldots, r\}} {\nmitte{\t}} w_{\t} \leq q(\nvector) \leq \frac{2}{3}\cdot
\sum_{\t\in \{1, \ldots, r\}} {\nmitte{\t}} w_{\t}$
(see Lemma~\ref{lemma:unweighted_median_per_type}). And our premises guarantee that the relative weight of each individual representative in $I_m$ shrinks to zero. The ``Main Theorem$^*$'' in \citeN{Neyman:1982}, therefore, has the following corollary:

\begin{lemma}[Neyman 1982]
  \label{lemma:own_neyman}
Given that $\nvector\in \mathcal{K}$,
\begin{equation}\label{eq:Neyman}
   \phi_\theta(\nvector) 
    =\frac{\ntI w_\t}{\sum_{{\t'}=1}^r \nmitte{\t'} w_{\t'}}\cdot(1+\mu(m))
    \quad \mbox{with} \quad \lim_{m\to \infty}|\mu(m)|=0.
\end{equation}
\end{lemma}
\begin{proof}
Neyman's theorem considers an infinite sequence of weighted voting games $[q^n; \boldsymbol{w}^n]$ with $n$ voters whose individual relative weights $w^n_i$ approach 0, and in which the relative quota $q^n$ is bounded away from 0 and 100\% (or at least $\lim_{n\to \infty} q^n/(\max_i w^n_i)=\infty$). Neyman establishes that\footnote{We somewhat specialize his finding and adapt the notation.}
\begin{equation}
\lim_{n\to \infty}| \phi_{T_n}(q^n;\boldsymbol{w}^n) - \sum_{i\in T_n} w^n_i| =0
\end{equation}
holds for any sequence of voter subsets $T_n\subseteq \{1,\ldots,n\}$, where $\phi_{T_n}(q^n;\boldsymbol{w}^n)$ denotes their aggregate Shapley value. (We here consider  $q^n=q(\nvector)/w_{\Sigma}$, $\boldsymbol{w}^n=(w_1, w_2, \ldots, w_k)/w_{\Sigma}$ and $T_n=\{i\in N \colon \tau(i)=\t \}$ for $N=\{1, \ldots, k\}$ and $w_{\Sigma}=\sum_{i\in N} w_i$.\footnote{Our notation leaves some inessential technicalities implicit: $\mathcal{K}$ really refers to a family of such sets, parameterized by $m$; we implicitly consider a sequence of $\nvector$-vectors such that $n=k\to\infty$ as $m\to \infty$.})

It is trivial that (\ref{eq:Neyman}) holds if $w_\t=0=\phi_\t(\nvector)$. So we can assume $w_\t>0$, and because there is at least the proportion $\beta>0$ of representatives from each type in $I_m$ for large $m$, the aggregate relative weight of $\t$-type representatives in $I_m$ is bounded away from 0, i.e.,\footnote{The limit itself need not exist because our premises do not rule out that, e.g., $m_\t$ is periodic in $m$.}
\begin{equation}\label{eq:TypRelativgewichtPositiv}
{\lim\inf}_{m\to \infty} \frac{\ntI w_\t}{\sum_{\t'=1}^r \nmitte{\t'}w_{\t'}}>0.
\end{equation}
Therefore, not only the absolute error $\tilde \mu(m)$ made in approximating $\phi_\t(\nvector)=\phi_{T_n}(q^n;\boldsymbol{w}^n)$ by $\frac{\ntI w_\t}{\sum_{\t'=1}^r \nmitte{\t'}w_{\t'}}$ but also the relative error $\mu(m)\equiv \tilde \mu(m) / \frac{\ntI w_\t}{\sum_{\t'=1}^r \nmitte{\t'}w_{\t'}}$ must vanish as $m\to \infty$.
\end{proof}

\subsubsection*{Step 5: Attributing aggregate pivot probabilities to individual representatives}

It then remains to disaggregate the pivot probabilities $\pi^\t(\mathcal{R}^{m})$ and $\pi^{\t'}(\mathcal{R}^{m})$ of types $\t$ and $\t'$ to individual representatives $i$ and $j$.
The aggregate relative weight of type-$\t$ representatives in the essential interval satisfies
\begin{equation}\label{eq:TypRelativgewichtsabsch}
    \frac{\ntI w_\t}{\sum_{{\t'}=1}^r \nmitte{\t'} w_{\t'}}=
    \frac{\beta_\t(m)m \ptI w_\t (1+ O(m^{-\frac{2}{5}}))}{\sum_{{\t'}=1}^r \beta_{\t'}(m)m \ptstrichI w_{\t'}(1-O(m^{-\frac{2}{5}}))}=
    \frac{\beta_\t(m) \ptI w_\t }{\sum_{{\t'}=1}^r \beta_{\t'}(m) \ptstrichI w_{\t'}}\left(1+ O(m^{-\frac{2}{5}})\right)
\end{equation}
for any $\nvector\in\mathcal{K}$ (see (II) in Lemma~\ref{lemma:concentration_per_type}).\footnote{To see the second equality note that for $y\in (0,\frac{1}{2})$ 
we have $\frac{1}{1-y}=1+y+y^2+\ldots\leq 1+2y=1+O(y)$. Similarly, $\frac{1}{1-y}\geq 1+y=1+O(y)$ and so $\frac{1}{1-y}=1+O(y)$.}
Combining this with equations (\ref{eq:TypSSIgesplitted}), (\ref{eq:cond_pivot_prob_as_expected_Shapley}) and (\ref{eq:Neyman}) yields
\begin{equation}\label{eq:ResultatTypebene}
    \lim_{m\to \infty}\frac{\pi^\t(\mathcal{R}^{m})}{\pi^{\t'}(\mathcal{R}^{m})}=
    \lim_{m\to \infty}\frac{\beta_\t(m) \ptI w_\t}{\beta_{\t'}(m) \ptstrichI w_{\t'}}=\lim_{m\to \infty}\frac{\beta_\t(m) f_{\t}(\median) w_\t}{\beta_{\t'}(m) f_{\t'}(\median) w_{\t'}}
\end{equation}
for arbitrary $\t,\t'\in \{1, \ldots,r\}$. Here, the final equality uses
\begin{equation}
    \lim_{m\to\infty} \frac{\ptI}{\ptstrichI}=
    \lim_{m\to\infty} \frac{\int_{-\eps(m)}^{\eps(m)}f_\t(x)dx}{\int_{-\eps(m)}^{\eps(m)}f_{\t'}(x)dx} =\frac{f_\t(\median)}{f_{\t'}(\median)},
\end{equation}
which can be deduced from (\ref{eq:p_t_bounding}).

\enlargethispage*{\baselineskip}

Our main result then follows from noting that the $m_\t=\beta_\t(m)\cdot m$ representatives of type $\t$ in assembly $\mathcal{R}^{m}$ are symmetric to each other and, therefore, must have identical pivot probabilities in $\mathcal{R}^{m}$. Hence
\begin{equation}\label{eq:ResultatIndividualebene}
\lim_{m\to \infty}\frac{\pi_i(\mathcal{R}^{m})}{\pi_j(\mathcal{R}^{m})}=
\lim_{m\to \infty}\frac{\pi^{\tau(i)}(\mathcal{R}^{m})/\beta_{\tau(i)}(m)} {\pi^{\tau(j)}(\mathcal{R}^{m})/\beta_{\tau(j)}(m)}=\frac{f_i(\median) w_i}{ f_j(\median) w_j}.
\end{equation}
\hfill $\blacksquare$
\bigskip
\bigskip

\subsection{Remarks}

Let us end this appendix with remarks on possible further generalizations.
First, the quadratic bound on $f_\t$'s variation in a neighborhood of $\median$ could be relaxed by choosing different constants in equations (\ref{eq:epsmdef}) and (\ref{eq:Hoeffding01}): $t(m_\t)=m_\t^{-b_1}$ with $b_1<\frac{1}{2}$ is all that is needed in order to ensure a vanishing error probability in (\ref{eq:Hoeffding01}); and $\eps(m)=m^{-b_2}$ with $b_2<b_1$ in (\ref{eq:epsmdef}) is sufficient for $\eps(m)\gg t(m_\t)$.
Then a local bound $|f_\t(x)-f_\t(\median)|\leq cx^a$ for $a>\frac{1-b_2}{b_2}$ is sufficient to establish Lemma~\ref{lemma:permutation_equiprobability}.
Requirement $b_2<b_1<\frac{1}{2}$ leaves generous room for $a<2$, but implies $a>1$.

Second, it is actually sufficient to assume \emph{local continuity} of all $f_\t$ at $\median$, rather than any strengthening of this,\footnote{Local continuity of $f_\t$ is obviously necessary: a modification of $f_\t(\median)$~-- with $f_\t(x)$ unchanged for $x\neq \median$~-- would affect $w_i f_\t(\median)$ but not $\pi_i(\mathcal{R}^{m})$.
Also the requirement of \emph{positive density} at the common median cannot be relaxed.
This is seen, e.g., by considering densities $f_i, f_j$ where $f_i(x)=0$ on a neighborhood $N_\eps(\median)$ while $f_j(\median)=0$ with $f_j(x)>0$ for $x\in N_\eps(\median)\setminus \{\median\}$; then $\pi_i(\mathcal{R}^{m})/\pi_j(\mathcal{R}^{m})$ converges to 0 rather than $w_i/w_j$.} if one appeals to an unpublished result by Abraham Neyman.
When, as in our setting, all voting weights have the same order of magnitude, the uniform convergence theorem of \citeN{Neyman:1982} for the Shapley value can be generalized to hold for all \emph{random order values} that are `sufficiently close' to the Shapley value.
More specifically, consider the expected marginal contribution of a voter~$i\in \{1, \ldots, k\}$
\begin{equation}\label{eq:random_order_value}
\Phi_i(v)\equiv \sum_{\ordering\in \mathcal{S}_k} p(\ordering)\cdot [v(T_i(\ordering) \cup \{i\})-v(T_i(\ordering))]
\end{equation}
in a weighted voting game $v=[q;w_1, \ldots, w_k]$, where any given permutation $\ordering\in\mathcal{S}_k$ on $N=\{1, \ldots, k\}$ has probability $p(\ordering)$, and $T_i(\ordering)\subset N$ denotes the set of $i$'s predecessors in $\ordering$, i.e., $T_i(\ordering)= \{j\colon \ordering(j)<\ordering(i)\}$. The random order value $\Phi(v)$ equals the Shapley value $\phi(v)$ if $p(\ordering)=\frac{1}{k!}$. This equiprobability can, for instance, be obtained by letting $\ordering$ be defined by the order statistics of a vector of random variables $\boldsymbol{X}=(X_1, \ldots, X_k)$ with mutually independent and $[0,1]$-uniformly distributed $X_1, \ldots, X_k$. The latter assumption can be relaxed somewhat without destroying the asymptotic proportionality of $i$'s weight $w_i$ and  $\Phi_i(v)$ which \citeN{Neyman:1982} has established when $\Phi(v)=\phi(v)$:

\begin{theorem}[Neyman, personal communication]\label{thm:Neyman_theorem_extension}
Fix $L>1$. For every $\varepsilon>0$ there exist $\delta>0$ and $K>0$ such that if $v$ is the weighted voting game $v=[q; w_1,\ldots, w_k]$ with $w_1, \ldots, w_k>0$, $\sum_{i=1}^k w_i=1$, $K\cdot \max_i w_i<q<1-K\cdot \max_i w_i$, $\max_{i,j} w_i/w_j<L$, and $\{p(\ordering)\}_{\ordering\in \mathcal{S}_k}$ in (\ref{eq:random_order_value}) is defined by the order statistics of independent $[0,1]$-valued random variables $X_1, \ldots, X_k$ with densities $f_i$ such that
$1-\delta<f_i(x)<1+\delta$ for every $x\in [0,1]$ and $i\in \{1,\ldots,k\}$ then \begin{equation}\label{eq:Neyman_theorem_extension}
    \sum_{i=1}^k |w_i-\Phi_i(v)| < \varepsilon.
\end{equation}
\end{theorem}
Of course, one can equivalently let $\{p(\ordering)\}_{\ordering\in \mathcal{S}_k}$
be defined by the order statistics of independent $I_m$-valued random variables with densities $\hat f_1, \ldots, \hat f_k$, instead of $[0,1]$-valued ones, if the theorem's condition $1-\delta<f_i(x)<1+\delta$ is replaced by the requirement that $\frac{1-\delta}{2\eps(m)} <\hat f_i(x)< \frac{1+\delta}{2\eps(m)}$ for all $x\in I_m$.

The values of $\delta$ and $L$ which one obtains for a given $\eps$
in Theorem~\ref{thm:Neyman_theorem_extension} apply to \emph{any} value of $k$.
We consider the weighted voting subgames played by the $k=\sum_{\t\in\{1,\dots,r\}} \nmitte{\t}$
representatives with realizations $\l_i\in I_m$ for given $\nvector\in \mathcal{K}$. The relative weight of any such representative~$i$,
$\hat w_i=w_i/\sum_{\t\in \{1,\dots,r\}} \nmitte{\t} w_\t$,
approaches zero as $m\to \infty$; and so does the maximum relative weight. Recalling that
the corresponding subgame's relative quota $\hat q=q(\boldsymbol{k})/\sum_{\t\in \{1,\dots,r\}} \nmitte{\t} w_\t$ is bounded by $\frac{1}{3}\le \hat q \le \frac{2}{3}$,
the condition $K\cdot \max_i \hat w_i < \hat q < 1-K\cdot \max_i \hat w_i$ is satisfied when $m$ is  sufficiently large. Any null players with $w_i=0$ can  w.l.o.g.\ be removed from consideration. Then all weights have the same order of magnitude, i.e., the choice of $L$ such that $\max_{i,j} \hat w_i/\hat w_j<L$ holds for all $\nvector\in \mathcal{K}$ is trivial.

Moreover, the conditional densities $\hat f_\t$ in our setup satisfy $\frac{1-\delta}{2\eps(m)} <\hat f_i(x)< \frac{1+\delta}{2\eps(m)}$ for every $\t\in \{1,\ldots, r\}$ and $x\in I_m$ when $m$ is large enough. Specifically, continuity of $f_\t$ in a neighborhood of $\median$ implies that for any given $\eps>0$ there exists $\Delta(\eps)>0$ with $\lim_{\eps\downarrow 0}\Delta(\eps)=0$ such that
\begin{align}
(1-\Delta(\varepsilon))\cdot{f}_{\t}(\median) \leq {f}_{\t}(x)\leq (1+\Delta(\varepsilon))\cdot{f}_{\t}(\median)
\end{align}
for all $x\in [\median-\eps,\median+\eps]$ and all $\t\in\{1,\dots,r\}$ (cf.\ inequality~(\ref{ie:bounds_f})).
Similarly to inequality~(\ref{eq:ptheta_bound_f}) we then conclude
\begin{align}
(1-\Delta(\varepsilon))f_\t(\median)\cdot 2\varepsilon \le \ptI \le (1+\Delta(\varepsilon))f_\t(\median)\cdot 2 \varepsilon.
\end{align}
Combining the last two inequalities with inequality~(\ref{eq_conditional_densitiy}) yields
\begin{align}
\frac{(1-\Delta(\varepsilon))}{(1+\Delta(\varepsilon))\cdot 2\eps} \le \hat f_{\t}(x)\le \frac{(1+\Delta(\varepsilon))}{(1-\Delta(\varepsilon))\cdot 2\eps}.
\end{align}
So considering $\eps=\eps(m)$ and any fixed $\delta$, the conditional densities $\hat f_\t$ satisfy $\frac{1-\delta}{2\eps(m)} <\hat f_i(x)< \frac{1+\delta}{2\eps(m)}$ for every $\t\in \{1,\ldots, r\}$ and $x\in I_m$ when $m$ is sufficiently large.

Hence, all premises in Neyman's unpublished Theorem~\ref{thm:Neyman_theorem_extension} are satisfied by the corresponding weighted voting subgames of agents with ideal points in $I_m$. Theorem~\ref{thm:Neyman_theorem_extension}, therefore, ensures the approximate weight proportionality of the aggregate random order value $\Phi$ of the type-$\t$ representatives.
Now if one recalls (\ref{eq:cond_pivot_prob_per_type})
and notices that the bracketed sum equals $\Phi(v)$ with $v=[\hat q; \hat w_{j_1}, \ldots, \hat w_{j_k}]$ when $j_1, \ldots, j_k$ denote the representatives with ideal points in $I_m$, we can replace Lemmata~\ref{lemma:permutation_equiprobability}--\ref{lemma:own_neyman} by the following:
\begin{lemma}\label{lemma:lemma4_5_combination}
\begin{equation}\label{eq:lemma4_5_combination}
\pi^\t(\mathcal{R}^{m}|\mathcal{K})=\frac{\ntI w_\t}{\sum_{{\t'}=1}^r \nmitte{\t'} w_{\t'}}\cdot(1+\mu(m))
    \quad \mbox{with} \quad \lim_{m\to \infty}|\mu(m)|=0.
\end{equation}
\end{lemma}

The proof of Theorem~\ref{thm:main_theorem} can then be concluded by appealing to (\ref{eq:TypSSIgesplitted}), hence
\begin{equation}
\lim_{m\to \infty}\frac{\pi^\t(\mathcal{R}^{m}|\mathcal{K})}{\pi^{\t'}(\mathcal{R}^{m}|\mathcal{K})}=
\lim_{m\to \infty}\frac{\pi^\t(\mathcal{R}^{m})}{\pi^{\t'}(\mathcal{R}^{m})},
\end{equation}
and equations (\ref{eq:ResultatTypebene})--(\ref{eq:ResultatIndividualebene}).
Importantly, the presumption $|f_\t(x)-f_\t(\median)|\leq cx^2$ for $x\in [\median-\eps_1,\median+\eps_1]$, which Lemma~\ref{lemma:permutation_equiprobability} required, is \emph{not} needed by Lemma~\ref{lemma:lemma4_5_combination}. It can hence be replaced in Theorem~\ref{thm:main_theorem} by the simpler requirement that each $f_\t$ is continuous in a neighborhood of $\median$.

Finally, the assumption that only a finite number of different densities and weights are involved in the chain $\mathcal{R}^{1}\subset \mathcal{R}^{2} \subset \mathcal{R}^{3}\subset \ldots$ could be loosened. However, it is critical that each representative's relative weight vanishes as $m\to \infty$ in order to apply Neyman's results; the asymptotic relation (\ref{eq:thm}) fails to hold, for instance, for a chain with  $w_1=\sum_{j>1}w_j$.
And because our result depends on a vanishing \emph{relative} error, which is considered neither by \citeN{Neyman:1982} nor Theorem~\ref{thm:Neyman_theorem_extension},\footnote{See, however, \citeN{Lindner/Machover:2004}, where conditions very similar to ours are considered for the Shapley and Banzhaf values, and the related discussion by \citeN{Lindner/Owen:2007}.} it is similarly important that the aggregate relative weight of each type of representatives is bounded away from zero. For instance, with just one representative having weight $w_1=1$ and $\beta_2(m)=m-1$ ones with $w_2=2$ (see equation (\ref{eq:1222_example})), $\lim_{m\to \infty}\pi_1(\mathcal{R}^{m})=\lim_{m\to \infty}\pi_j(\mathcal{R}^{m})=0$ for any $j\neq 1$ but the limit of ${\pi_1(\mathcal{R}^{m})}/{\pi_j(\mathcal{R}^{m})}$ may fail to exist. 

\section{Proof of Theorem~\ref{thm:main_result_for_shock_case}}\label{app:shockcase}

\begin{theoremshock*}
Consider an assembly $\mathcal{R}^{m,q}$ with an arbitrary number $m$ of constituencies and the relative decision quota $q\in [0.5; 1)$. For each $i\in \{1, \ldots, m\}$ let $\l_i=t\cdot \mu_i+\tilde \epsilon_i$, where
$\mu_1,\ldots, \mu_m$ and $\tilde \epsilon_1,\ldots,\tilde \epsilon_m$ are all mutually independent random variables,
$\tilde \epsilon_1,\ldots,\tilde \epsilon_m$ have finite means and variances, and
$\mu_1,\ldots, \mu_m$ have an identical bounded density.
Then
\begin{equation}\tag{\ref{eq:shockthm}}
\lim_{\sigext\to \infty}\frac{\pi_i(\mathcal{R}^{m,q,t})}{\pi_j(\mathcal{R}^{m,q,t})} = \frac{\phi_i(v)}{\phi_j(v)}.
\end{equation}
\end{theoremshock*}

The result easily follows from the definition of the Shapley value and the fact that the orderings which are induced by the realizations of the vectors $\boldsymbol{\l}=(\l_1, \ldots, \l_m)$ and $\boldsymbol{\mu}=(\mu_1, \ldots, \mu_m)$
will coincide with a probability which tends to $1$ as $t$ approaches infinity. To see the latter, ignore any null events in which several ideal points or constituency shocks coincide and let $\hat\ordering(\mathbf{x})$ denote the permutation
of $\{1,\dots,m\}$ such that $x_i<x_j$ whenever $\hat\ordering(i)<\hat\ordering(j)$ for the real-valued vector $\mathbf{x}=(x_i)_{i\in\{1,\dots,m\} }$. We then have:

\begin{lemma}
For $i\in\{1, \ldots, m\}$ and $t > 0$ let $\l_i^t\equiv t\cdot \mu_i+\tilde \epsilon_i$, where
$\mu_1,\ldots, \mu_m$ and $\tilde \epsilon_1,\ldots,\tilde \epsilon_m$ are all mutually independent random variables,
$\tilde \epsilon_1,\ldots,\tilde \epsilon_m$ have finite means and variances, and
$\mu_1,\ldots, \mu_m$ have an identical bounded density.
Then
\begin{equation}
  \lim\limits_{t\to\infty} \Pr(\hat\ordering(\l^t)=\ordering)=\lim\limits_{t\to\infty} \Pr(\hat\ordering(\mu)=\ordering)=\frac{1}{m!}
\end{equation}
for each permutation $\ordering$ of $\{1,\dots,m\}$.
\end{lemma}

\begin{proof}
Let us denote the finite variance of $\tilde \epsilon_i$ by $\sigma_i^2$ and let
$U\equiv \left(\max_i |\E[\tilde\epsilon_i]|\right)^3$.
We can choose a real number $k$ such that the bounded density function $h$ of $\mu_i$, with $i\in \{1, \ldots, m\}$,
satisfies $h(x)\le k$ for all $x\in\mathbb{R}$. For any given realization $\mu_j=x$, the probability of the independent
random variable $\mu_i$ assuming a value inside interval $(x-4t^{-\frac{2}{3}}, x+4t^{-\frac{2}{3}})$
is bounded above by $k\cdot 8 t^{-\frac{2}{3}}$. We can infer that the event
$\big\{|\mu_i-\mu_j|<4 t^{-\frac{2}{3}}\big\}$, which is identical to the event
$\big\{|t\mu_i-t\mu_j|<4 t^{\frac{1}{3}}\big\}$, has a probability of at most $k\cdot 8 t^{-\frac{2}{3}}$
for any $i\neq j\in \{1, \ldots, m\}$. And we can conclude from Chebyshev's inequality that
$\Pr(|\tilde\epsilon_i-\E[\tilde\epsilon_i]|< t^{\frac{1}{3}})$ is at least
$1-\sigma_i^2\cdot t^{-\frac{2}{3}}$.
For $t\ge U$, we have $|\E[\tilde\epsilon_i]|\le t^{\frac{1}{3}}$; and if $|\tilde\epsilon_i-\E[\tilde\epsilon_i]|< t^{\frac{1}{3}}$ holds then also
\begin{equation}\label{ie:absolute_of_epsilon_bound}
2t^{\frac{1}{3}}>|\E[\tilde\epsilon_i]| + |\tilde\epsilon_i-\E[\tilde\epsilon_i]| \ge |\tilde\epsilon_i|
\end{equation}
by the triangle inequality. Hence, the probability for (\ref{ie:absolute_of_epsilon_bound}) to hold when $t\ge U$ is $\Pr(|\tilde\epsilon_i|< 2t^{\frac{1}{3}})\ge 1-\sigma_i^2\cdot t^{-\frac{2}{3}}$ for each $i\in \{1, \ldots, m\}$.

Now consider the joint event that (i) $|t\mu_i-t\mu_j|\ge 4t^{\frac{1}{3}}$ for \emph{all} pairs
$i\neq j\in \{1, \ldots, m\}$ and (ii) that $|\tilde\epsilon_i|< 2t^{\frac{1}{3}}$ for \emph{all}
$i\in \{1, \ldots, m\}$. In this event, the ordering of $\l_1^t, \ldots, \l_m^t$ is determined entirely by the
realization of $t \mu_1, \ldots, t \mu_m$; in particular, $\hat\ordering(\boldsymbol{\l}^t)=\hat\ordering(\boldsymbol{\mu})$.
Using the mutual independence of the considered random variables this joint event must have a probability of at least
\begin{equation}
  \prod\limits_{s=1}^{m \choose 2} \left(1-k\cdot 8 t^{-\frac{2}{3}}\right) \cdot \prod\limits_{i=1}^m
  \left(1-\sigma_i^2\cdot t^{-\frac{2}{3}}\right)
  \ge 1- \left(8 k{m \choose 2}+\sum_{i=1}^m \sigma_i^2\right)\cdot t^{-\frac{2}{3}}
\end{equation}
for $t\ge U$.
The right hand side clearly tends to $1$ as $t$ approaches infinity. It hence remains to acknowledge that any
ordering $\hat\ordering(\boldsymbol{\mu})$ has an equal probability of $1/m!$ because $\mu_1, \ldots, \mu_m$ are i.i.d.
\end{proof}

\setlength{\labelsep}{-0.2cm}


\newcommand{\noopsort}[1]{}

\end{document}